\documentclass[12pt,preprint,sort&compress,longtitle]{elsarticle}

\usepackage{pdfsync}
\usepackage{amssymb}
\usepackage{amsmath}
\usepackage[caption=false]{subfig}
\usepackage{multirow}
\usepackage{algorithmic}
\usepackage{algorithm}
\usepackage{amsthm}
\usepackage{relsize}
\usepackage{url}
\usepackage[table]{xcolor}

\newcommand{\defword}[1]{\textbf{\boldmath{#1}}}

\newcommand{\cI}{\mathcal{I}}

\newcommand{\cS}{\mathcal{S}}
\newcommand{\RR}{\mathbb{R}}
\newcommand{\ISTATE}[1][1]{\STATE\hspace{#1\algorithmicindent}}

\newtheorem{definition}{Definition}
\newtheorem{theorem}{Theorem}
\newtheorem{lemma}{Lemma}
\newtheorem{prop}{Proposition}
\newtheorem{corollary}{Corollary}
\newtheorem*{prop-weakdom}{Proposition \ref{prop:weakdom}}
\newtheorem*{theorem-normaldom}{Theorem \ref{thm:normaldom}}
\newtheorem*{theorem-domacts}{Theorem \ref{thm:domacts}}
\newtheorem*{theorem-domstrats}{Theorem \ref{thm:domstrats}}
\newproof{proof-sketch}{Proof Sketch}
\newproof{proof-prop-weakdom}{Proof of Proposition \ref{prop:weakdom}}
\newproof{proof-thm-normaldom}{Proof of Theorem \ref{thm:normaldom}}
\newproof{proof-thm-domacts}{Proof of Theorem \ref{thm:domacts}}
\newproof{proof-thm-domstrats}{Proof of Theorem \ref{thm:domstrats}}

\def\etal{\textit{et al.}}
\def\ie{\textit{i.e.}}

\journal{Artificial Intelligence}

\begin{document}

\begin{frontmatter}



\title{Regret Minimization in Non-Zero-Sum Games with Applications to Building Champion Multiplayer Computer Poker Agents}


\author{Richard Gibson}
\ead{rggibson@cs.ualberta.ca}
\ead[url]{http://cs.ualberta.ca/~rggibson/}

\address{Department of Computing Science, University of Alberta, 2-21 Athabasca Hall, Edmonton, AB, Canada, T6G 2E8. Tel: 1-780-492-2821}

\begin{abstract}
In two-player zero-sum games, if both players minimize their average external regret, then the average of the strategy profiles converges to a Nash equilibrium. 
For $n$-player general-sum games, however, theoretical guarantees for regret minimization are less understood. 
Nonetheless, Counterfactual Regret Minimization (CFR), a popular regret minimization algorithm for extensive-form games, has generated winning three-player Texas Hold'em agents in the Annual Computer Poker Competition (ACPC). 
In this paper, we provide the first set of theoretical properties for regret minimization algorithms in non-zero-sum games by proving that solutions eliminate iterative strict domination. 
We formally define \emph{dominated actions} in extensive-form games, show that CFR avoids iteratively strictly dominated actions and strategies, and demonstrate that removing iteratively dominated actions is enough to win a mock tournament in a small poker game. 
In addition, for two-player non-zero-sum games, we bound the worst case performance and show that in practice, regret minimization can yield strategies very close to equilibrium. 
Our theoretical advancements lead us to a new modification of CFR for games with more than two players that is more efficient and may be used to generate stronger strategies than previously possible. 
Furthermore, we present a new three-player Texas Hold'em poker agent that was built using CFR and a novel game decomposition method.  
Our new agent wins the three-player events of the 2012 ACPC and defeats the winning three-player programs from previous competitions while requiring less resources to generate than the 2011 winner. 
Finally, we show that our CFR modification computes a strategy of equal quality to our new agent in a quarter of the time of standard CFR using half the memory. 
\end{abstract}


\end{frontmatter}

\maketitle
\allowdisplaybreaks
\sloppy



\section{Introduction}
\label{sec:intro}


Normal-form games are a common and general framework useful for modelling problems involving single, simultaneous decisions made by multiple agents. 
When decisions are sequential and involve imperfect information or stochastic events, extensive-form games are generally more practical. 

A common solution concept in games is a Nash equilibrium strategy profile that guarantees no player can gain utility by unilaterally deviating from the profile. 
For two-player zero-sum games, a Nash equilibrium is a powerful notion. 
In such domains, every Nash equilibrium profile results in the players earning their unique \emph{game value}, and playing a strategy belonging to a Nash equilibrium guarantees a payoff no worse than the game value. 
In $n$-player general-sum games, these strong guarantees are lost. 
Each Nash equilibrium may provide different payoffs to the players and no guarantee can be made when more than one player deviates from a specific equilibrium profile. 
Regardless, no practical algorithms are known for computing an equilibrium in even moderately-sized games with more than two players. 

Counterfactual Regret Minimization (CFR) \cite{ZinkevichEtAl2008} is a state-of-the-art algorithm for approximating Nash equilbria of large two-player zero-sum extensive-form games. 
CFR is an iterative, off-line regret minimizer that stores two strategy profiles, the \emph{current profile} that is being played at the present iteration, and the \emph{average profile} that accumulates a running average of all previous profiles generated. 
In two-player zero-sum games, the average profile approaches a Nash equilibrium and is generally used in practice, while the current profile is discarded. 
CFR can also be applied to non-zero-sum games and games with more than two players, but the average profile does not necessarily approximate an equilibrium in such cases \cite[Table 2]{AAMAS2010}. 
Previous work provides no theoretical insights into the average profile outside of two-player zero-sum games. 

Nonetheless, CFR has been applied successfully to games that are not two-player zero-sum.
For example, CFR was used to generate more aggressive, or \emph{tilted}, poker strategies from non-zero-sum games capable of defeating top poker professionals in two-player limit Texas Hold'em \cite{rgbr}. 
In addition, winning three-player Texas Hold'em poker strategies in the Annual Computer Poker Competition (ACPC) \cite{ACPC} have been constructed using CFR \cite{AAMAS2010,stitching}.  
As CFR's memory requirements are linear in the size of the game, a common approach in poker is to employ a state-space abstraction that merges different card deals into \emph{buckets}, leaving hands in the same bucket indistinguishable \cite{Gilpin&Sandholm2006,Abstraction2013}. 
Three-player limit Texas Hold'em contains over $10^{17}$ decision states, and so many hands must be merged for CFR to be feasible. 
In 2011, the winning three-player agent combated this problem through \emph{heads-up expert strategies} \cite{AAMAS2010} that merged fewer hands and only acted in common two-player scenarios resulting from one player folding early in a hand. 
While CFR has been successful in these games, a reason why CFR might be successful in such domains has not been given. 

In this paper, we provide the first theoretical groundings for regret minimization algorithms applied to games that are not two-player zero-sum. 
This is achieved by establishing elimination of iteratively dominated errors: mistakes where there exists an alternative that is guaranteed to do better, assuming the opponents do not make such errors themselves. 
The contributions of this paper are as follows. 
Firstly, we prove that in normal-form games, common regret minimization techniques eliminate (play with probability zero) iteratively strictly dominated strategies. 
Secondly, we formally define \emph{dominated actions} and prove that under certain conditions, CFR eliminates iteratively strictly dominated actions and strategies. 
Thirdly, for two-player non-zero-sum games, we bound the average profile's worst-case performance, providing a theoretical understanding of tilted poker strategies. 
Fourthly, our theoretical results lead us to a simple modification of CFR for games with more than two players that only uses the current profile and does not average. 
We demonstrate that with this change, CFR generates poker strategies that perform just as well as those generated without the change, but now require less time and less memory to compute. 
Furthermore, for large games requiring state-space abstraction, this reduction in memory allows finer-grained abstractions to be used by CFR, leading to even stronger strategies than previously possible.
Fifthly, we develop a new three-player limit Texas Hold'em agent that, instead of using heads-up experts, varies its abstraction quality according to the estimated \emph{importance} of each state. 
Our new agent wins the three-player events of the 2012 ACPC and defeats the previous years' champions, all while needing less computer memory to generate than the 2011 winner.  

The rest of this paper is organized as follows.  
Section \ref{sec:background} covers background material in game theory and solution concepts relevant to our work. 
Next, Section \ref{sec:regret} discusses regret minimization and provides an overview of CFR in extensive-form games. 
We then formally define dominated actions in Section \ref{sec:dom-actions} before proving our theoretical results in Section \ref{sec:theory}. 
Section \ref{sec:experiments} explores these theoretical findings and insights empirically across a number of different poker games. 
Our new champion three-player Texas Hold'em agent is then described and evaluated in Section \ref{sec:acpc}. 
Finally, Section \ref{sec:conclusion} concludes our work and discusses future research directions. 
Proof sketches are provided with the theorem statements, while full technical proofs are provided in \ref{sec:appendix}. 

\section{Games}
\label{sec:background}


\subsection{Normal and Extensive Forms} 

A finite \defword{normal-form game} is a tuple $G = \langle N, A, u \rangle$ where $N = \{1, ..., n\}$ is the set of \defword{players}, $A = A_1 \times \dots \times A_n$ is the set of \defword{action profiles} with $A_i$ being the finite set of \defword{actions} available to player $i$, and $u_i: A \rightarrow \RR$ is the \defword{utility function} that denotes the payoff for player $i$ at each possible action profile. 
If $n = 2$ and $u_1 = -u_2$, the game is \defword{two-player zero-sum} (or simply \defword{zero-sum}). 
Otherwise, the game is \defword{non-zero-sum}. 
Two-player normal-form games are often represented by a matrix with rows denoting the row player's actions, columns denoting the column player's actions, and entries indicating utilities resulting from the row player's and column player's actions respectively. 
A \defword{mixed strategy} $\sigma_i$ for player $i$ is a probability distribution over $A_i$, where $\sigma_i(a)$ is the probability that action $a$ is taken under $\sigma_i$. 
The set of all such strategies for player $i$ is denoted $\Sigma_i$.  
Define the \defword{support} of $\sigma_i$, $\text{supp}(\sigma_i)$, to be the set of actions assigned positive probability by $\sigma_i$. 
A \defword{strategy profile} $\sigma \in \Sigma$ is a collection of strategies $\sigma = (\sigma_1, ..., \sigma_n)$, one for each player. 
We let $\sigma_{-i}$ refer to the strategies in $\sigma$ excluding $\sigma_i$, and $u_i(\sigma)$ to be the expected utility for player $i$ when players play according to $\sigma$. 


Extensive-form games are often preferred to normal form when multiple decisions are made sequentially. 
Before providing the formal definitions, we describe Kuhn Poker, an extensive-form game that we will use as a running example throughout this paper. 
Kuhn Poker \cite{Kuhn} is a zero-sum card game played with a three-card deck containing a Jack, Queen, and King. 
Each player antes one chip and is dealt one private card at random from the deck that no other player can see.
There is a single round of betting starting with player 1, who may either check or bet one chip. 
If a bet is made, player 2 can either fold and forfeit the hand, or call the one chip bet. 
When faced with a check, player 2 can either check or bet one chip, where a bet forces player 1 to either fold or call the bet. 
If neither player folds after the round of betting, then the player with the highest ranked card wins all of the chips played. 

In general, a finite \defword{extensive-form game} with imperfect information \cite{Osborne&Rubenstein1994} is a tuple $\Gamma = \langle N, A, H, P, \sigma_c, u, \cI \rangle$ that contains a game tree with nodes corresponding to \defword{histories} of actions $h \in H$ and edges corresponding to \defword{actions} $a \in A(h)$ available to $\defword{player}$ $P(h) \in N \cup \{c\}$ (where again $N$ is the set of players and $c$ denotes \defword{chance}). 
For histories $h, h' \in H$, we call $h$ a \defword{prefix} of history $h'$, written $h \sqsubseteq h'$, if $h'$ begins with the sequence $h$. 
When $P(h) = c$, $\sigma_c(h,a)$ is the (fixed) probability of chance generating action $a$ at $h$. 
Terminal nodes correspond to \defword{terminal histories} $z \in Z \subseteq H$ that have associated \defword{utilities} $u_i(z)$ for each player $i$. 
We define $\Delta_i = \max_{z,z' \in Z} u_i(z) - u_i(z')$ to be the range of utilities for player $i$. 
Non-terminal histories for player $i$, $H_i$, are partitioned into \defword{information sets} $I \in \cI_i$ representing the different game states that player $i$ cannot distinguish between. 
For example, in Kuhn Poker, player $i$ does not see the private card dealt to the opponent, and thus every pair of histories differing only in the private card of the opponent are in the same information set for player $i$. 
For each $I \in \cI_i$, the action sets $A(h)$ must be identical for all $h \in I$, and we denote this set by $A(I)$. 
Define $|A(\cI_i)| = \max_{I \in \cI_i} |A(I)|$ to be the maximum number of actions available to player $i$ at any information set. 
We assume \defword{perfect recall} that guarantees players always remember information that was revealed to them, the order it was revealed, and the actions they chose. 

A \defword{behavioral strategy} for player $i$, $\sigma_i \in \Sigma_i$, is a function that maps each information set $I \in \cI_i$ to a probability distribution over $A(I)$. 
Denote $\pi^\sigma(h)$ as the probability of history $h$ occurring if all players play according to $\sigma=(\sigma_1,...,\sigma_n)$. 
We can decompose $\pi^\sigma(h) = \prod_{i \in N \cup \{c\}} \pi_i^\sigma(h)$ into each player's and chance's contribution to this probability. 
Here, $\pi_i^\sigma(h)$ is the contribution to this probability from player $i$ when playing according to $\sigma_i$. 
Let $\pi_{-i}^\sigma(h)$ be the product of all contributions (including chance) except that of player $i$. 
In addition, let $\pi^\sigma(h,h')$ be the probability of history $h'$ occurring after $h$, given that $h$ has occurred. 
Let $\pi_i^\sigma(h,h')$ and $\pi_{-i}^\sigma(h,h')$ be defined similarly. 
Furthermore, we define the probability of player $i$ reaching information set $I \in \cI_i$ as $\pi_i^\sigma(I) = \pi_i^\sigma(h)$ for any $h \in I$. 
This is well-defined due to perfect recall as any two histories reaching the same information set must have followed the same sequence of actions at previous, identical information sets. 

A strategy $s_i$ is \defword{pure} if a single action is assigned probability 1 at every information set; for each $I \in \cI_i$, let $s_i(I)$ be this action.  
Denote $\cS_i$ as the set of all pure strategies for player $i$. 
For a behavioral strategy $\sigma_i$, define the \defword{support} of $\sigma_i$ to be $\text{supp}(\sigma_i) = \{s_i \in \cS_i \mid \sigma_i(I, s_i(I)) > 0 \text{ for all } I \in \cI_i \}$. 
Note that normal form is a generalization of extensive form. 
An extensive-form game $\Gamma$ can be represented in normal form $G$ by setting the action set in $G$ for player $i$ to be the set of all pure strategies in $\Gamma$ and assigning utility $u_i(s) = \sum_{z \in Z} \pi^s(z) u_i(z)$. 
Then, every behavioral strategy $\sigma_i$ in $\Gamma$ has a utility-equivalent mixed strategy in $G$ where the probability of selecting $s_i$ is $\prod_{I \in \cI_i} \sigma_i(I, s_i(I))$ \cite{Kuhn53}. 
However, normal form is often impractical for even moderately-sized problems because the size of the action set in $G$ is exponential in $|\cI_i| \cdot |A(\cI_i)|$. 

\subsection{Solution Concepts}

In this paper, we consider the problem of computing a strategy profile to a game for play against a set of unknown opponents. 
The most common solution concept is the Nash equilibrium. 
For $\epsilon \geq 0$, a strategy profile $\sigma$ is an \defword{$\epsilon$-Nash equilibrium} if no player can unilaterally deviate from $\sigma$ and gain more than $\epsilon$; \ie, $\max_{\sigma_i' \in \Sigma_i} u_i(\sigma_i', \sigma_{-i}) \leq u_i(\sigma) + \epsilon$ for all $i \in N$. 
A $0$-Nash equilibrium is simply called a \defword{Nash equilibrium}. 
For games with more than two players, computing a Nash equilibrium is hard and belongs to the PPAD-complete class of problems \cite{chen:3-nash,ChenDeng2player,DaskalakisPapadimitriou,DaskalakisGoldbergPapadimitriou}.
Alternatively, we consider a superset of Nash equilibria, particularly those profiles that avoid iterative strict domination.
\begin{definition}
\label{def:domination}
A strategy $\sigma_i$ for player $i$ is a \defword{strictly dominated strategy} if there exists another player $i$ strategy $\sigma_i'$ such that $u_i(\sigma_i, \sigma_{-i}) < u_i(\sigma_i', \sigma_{-i})$ for all $\sigma_{-i} \in \Sigma_{-i}$.
\end{definition}
\noindent \defword{Weak} and \defword{very weak} dominance have also been studied that allow equality instead of strict inequality for all but one and for all opponent profiles respectively. 
For each type of dominance, an \defword{iteratively dominated strategy} is any strategy that is either dominated or becomes dominated after successively removing iteratively dominated strategies from the game. 
In this paper, we focus on strict domination where it is well-known that iterated removal of strictly dominated strategies always results in the same set of remaining strategies, regardless of the order of removal \cite{GilboaKalaiZemel}. 

Conitzer and Sandholm \cite{DominanceComplexity} prove that a strictly dominated strategy $\sigma_i \in \Sigma_i$ in a normal-form game can be identified in time polynomial in $|A_i| = |\cS_i|$ by showing that the objective of the linear program
\begin{align}
\textbf{minimize } & \sum_{s_i \in \cS_i} p_{s_i} \label{lp:dom} \\
\textbf{subject to } & \forall s_{-i} \in \cS_{-i}, \sum_{s_i \in \cS_i} p_{s_i} u_i(s_i, s_{-i}) \geq u_i(\sigma_i, s_{-i}) \nonumber
\end{align} 
is less than 1, where each $p_{s_i}$ is a nonnegative real number.  
Iteratively strictly dominated strategies can then be eliminated by repeatedly solving this program and removing the dominated pure strategies from $\cS_i$ and $\cS_{-i}$. 
However, this method is infeasible for large extensive-form games as the linear programs would require an exponential number of constraints in the size of the game. 
Hansen \etal~\cite{HansenDynamicProgramming} develop a dynamic programming algorithm for partially observable stochastic games, a generalization of normal-form games, that removes iteratively very weakly dominated strategies, but is not practical beyond small toy problems.  
Further insights are provided by Waugh's \emph{domination value} \cite{Waugh2009} that attempts to measure the amount of utility lost through playing iteratively dominated strategies in zero-sum games. 
Waugh demonstrates a strong correlation between the domination value of a strategy with performance in a small poker game, suggesting that removal of dominated strategies is enough for good play. 
This particular work motivates our results in Section \ref{sec:theory}. 

Two other generalizations of Nash equilibria, correlated and coarse correlated equilibria, require a mechanism for correlation among the players. 
Suppose an independent moderator selects a profile $\sigma^k$ from $E = \{\sigma^1, ..., \sigma^K\}$ according to distribution $q$ and privately recommends each player $i$ play strategy $\sigma_i^k$. 
Then $(E, q)$ is a \defword{correlated equilibrium} if no player has an incentive to unilaterally deviate from any recommendation. 
A \defword{coarse correlated equilibrium} is similar but even more general, where for all $i \in N$, we only require that
\begin{equation}
\label{eq:cce}
\sum_{k=1}^K q(k)u_i(\sigma^k) \geq \max_{\sigma_i' \in \Sigma_i} \sum_{k=1}^K q(k) u_i(\sigma_i', \sigma_{-i}^k). 
\end{equation}
To not be in a coarse correlated equilibrium, a player would need incentive to deviate even before receiving a recommendation and the deviation must be independent of the recommendation. 
Without a mechanism for correlation, it is unclear how a practitioner should use a correlated equilibrium. 
In addition, while correlated equilibria remove dominated strategies \cite{BlumMansour2007}, a coarse correlated equilibrium may lead to the recommendation of a strictly dominated strategy. 
For example, in the normal-form game in Figure \ref{fig:cce}, $\{(A,a) = 0.5, (B,b) = 0.25, (C,b) = 0.25 \}$ is a coarse correlated equilibrium with the row player's expected utility being $5/4$
, yet the strictly dominated row player strategy that always plays $C$ is recommended 25\% of the time. 

\begin{figure}%
\centering
$\bordermatrix{
	  & a   & b   \cr
	A & 1,0 & 0,0 \cr
	B & 0,0 & 2,0 \cr
	C & -1,0 & 1,0 \cr}$
\caption{A two-player non-zero-sum normal-form game, where the column player's utility is always zero.}%
\label{fig:cce}%
\end{figure}

\section{Regret Minimization}
\label{sec:regret}

Given a sequence of strategy profiles $\sigma^1, ..., \sigma^T$, the \defword{(external) regret} for player $i$ is
\[ R_i^T = \max_{\sigma_i' \in \Sigma_i} \sum_{t=1}^T \left( u_i( \sigma_i', \sigma_{-i}^t ) - u_i(\sigma^t) \right). \]
$R_i^T$ measures the amount of utility player $i$ could have gained by following the single best fixed strategy in hindsight at all time steps $t = 1,...,T$. 
Theorem \ref{thm:folk} below states a well-known result that relates regret to Nash equilibria in zero-sum games:
\begin{theorem}
\label{thm:folk}
In a zero-sum game, for $\epsilon \geq 0$, if $R_i^T \leq \epsilon$ for $i=1,2$, then the average strategy profile, $\bar{\sigma}^T$ (defined later), is a $2\epsilon$-Nash equilibrium. 
\end{theorem}
A proof is provided by Waugh \cite[p.~11]{Waugh2009}. 
It is also well-known that in any game, minimizing internal regret, a stronger notion of regret, leads to a correlated equilibrium, but we only consider external regret here. 

\subsection{Regret Matching and CFR}

\defword{Regret matching} \cite{RegretMatching} is a very simple, iterative procedure that minimizes average regret in a normal-form game. 
First, the initial profile $\sigma^1$ is chosen arbitrarily. 
For each action $a \in A_i$, we store the accumulated regret $R_i^T(a) = \sum_{t=1}^T \left( u_i(a, \sigma_{-i}^t) - u_i(\sigma_i^t, \sigma_{-i}^t) \right)$ that measures how much player $i$ would rather have played $a$ at each time step $t$ than follow $\sigma_i^t$. 
Successive strategies are then determined according to
\begin{equation}
\label{eq:RM}
\sigma_i^{T+1}(a) = \frac{R_i^{T,+}(a)}{\sum_{b \in A_i} R_i^{T,+}(b)},
\end{equation}
where $x^+ = \max\{x, 0\}$ and actions are chosen arbitrarily when the denominator is zero. 
One can show that
\begin{equation}
\label{eq:RMBound}
\frac{R_i^T}{T} = \max_{a \in A_i} \frac{R_i^T(a)}{T} \leq \frac{\Delta_i \sqrt{|A_i|}}{\sqrt{T}}.
\end{equation}
A general proof is provided by Gordon \cite{Gordon2007}, while a more direct proof is provided by Lanctot \cite[Theorem 2]{lanctot-thesis}. 
By Theorem \ref{thm:folk}, the \defword{average strategy profile}, defined by $\bar{\sigma}_i^T(a) = \sum_{t=1}^T \sigma_i^t(a) / T$, approaches a Nash equilibrium as $T \rightarrow \infty$. 

Regret matching requires storage of $R_i^T(a)$ for all $a \in A_i$.  
Thus, it is infeasible to directly apply regret matching to even moderately-sized extensive-form games due to the resulting exponential size of the action (pure strategy) space. 
Alternatively, \defword{Counterfactual Regret Minimization (CFR)} \cite{ZinkevichEtAl2008} is a state-of-the-art algorithm that minimizes average regret while only requiring storage proportional to $|\cI_i| \cdot |A(\cI_i)|$ in the extensive-form game. 
Pseudocode is provided in Algorithm \ref{alg:cfr}. 
On each iteration $t$ and for each player $i$, the expected utility for player $i$ is computed at each information set $I \in \cI_i$ under the current profile $\sigma^t$, assuming player $i$ plays to reach $I$. 
This expectation is the \defword{counterfactual value} for player $i$,
\[ v_i(I, \sigma) = \sum_{z \in Z_I} u_i(z) \pi_{-i}^\sigma(z[I]) \pi^\sigma(z[I], z), \]
where $Z_I$ is the set of terminal histories passing through $I$ and $z[I]$ is the history leading to $z$ contained in $I$. 
For each action $a \in A(I)$, these values determine the \defword{counterfactual regret} at iteration $t$, $r_i^t(I,a) = v_i(I, \sigma_{(I \rightarrow a)}^t) - v_i(I, \sigma^t)$, where $\sigma_{(I \rightarrow a)}$ is the profile $\sigma$ except at $I$, action $a$ is always taken. 
The regret $r_i^t(I,a)$ measures how much player $i$ would rather play action $a$ at $I$ than follow $\sigma_i^t$ at $I$. 
These regrets are accumulated to obtain the \defword{cumulative counterfactual regret}, $R_i^T(I,a) = \sum_{t=1}^T r_i^t(I,a)$, that define the \defword{current strategy profile} via regret matching at $I$, 
\begin{equation}
\label{eq:CFRRM}
\sigma_i^{T+1}(I, a) = \frac{R_i^{T,+}(I, a)}{\sum_{b \in A(I)} R_i^{T,+}(I,b)}.
\end{equation} 
This procedure minimizes average regret according to the bound
\begin{equation}
\label{eq:CFRBound}
\frac{R_i^T}{T} \leq \frac{\Delta_i |\cI_i| \sqrt{|A(\cI_i)|}}{\sqrt{T}} \text{ \cite[Theorem 4]{ZinkevichEtAl2008}}.
\end{equation}
During computation, CFR stores a \defword{cumulative profile} $s_i^T(I,a) = \sum_{t=1}^T \pi_i^{\sigma^t}(I) \sigma_i^t(I,a)$. 
Once CFR is terminated after $T$ iterations, the output is the \defword{average strategy profile} $\bar{\sigma}_i^T(I,a) = s_i^T(I,a) / \sum_{b \in A(I)} s_i^T(I,b)$. 


\begin{algorithm}[t]
  \caption{Counterfactual Regret Minimization (Zinkevich \etal~2008)}
  \label{alg:cfr}
\begin{algorithmic}[1]
  \STATE Initialize regret: $\forall I, a \in A(I):  R(I,a) \leftarrow 0$
  \STATE Initialize cumulative profile: $\forall I, a \in A(I): s(I,a) \leftarrow 0$
  \STATE Initialize current profile: $\forall I, a \in A(I): \sigma(I, a) = 1 / |A(I)|$ \label{alg:background:cfr:init-cur}
  \STATE {\bf for} $t \in \{1, 2, ..., T\}$ {\bf do}
  	\ISTATE[1] {\bf for} $i \in N$ {\bf do}
  		\ISTATE[2] {\bf for} $I \in \mathcal{I}_i$ {\bf do}
  			\ISTATE[3] $\sigma_i(I, \cdot) \leftarrow$ RegretMatching($R(I, \cdot)$)
  			\ISTATE[3] {\bf for} $a \in A(I)$ {\bf do}
  				\ISTATE[4] $R(I,a) \leftarrow R(I,a) + v_i(I, \sigma_{(I \rightarrow a)}) - v_i(I, \sigma)$
  				\ISTATE[4] $s(I,a) \leftarrow s(I,a) + \pi_i^\sigma(I) \sigma_i(I,a)$
  			\ISTATE[3] {\bf end for} 
  		\ISTATE[2] {\bf end for}
  	\ISTATE[1] {\bf end for}
  \STATE {\bf end for}
\end{algorithmic}
\end{algorithm} 

Since all players are minimizing average regret, it follows by Theorem \ref{thm:folk} that for zero-sum games, CFR's average profile converges to a Nash equilibrium. 
For non-zero-sum games, if we assign probability $1/T$ to each of the profiles $\{\sigma^1, ..., \sigma^T\}$ generated by CFR or any other regret minimizer, then by equation \eqref{eq:cce} and minimization of regret, this converges to a coarse correlated equilibrium. 
Though previous work omits this fact, it is unclear how this could be useful, let alone why the average strategy $\bar{\sigma}_i^T$ might be valuable. 
However, the average strategy has been shown to perform well empirically in non-zero-sum games against human opponents and competitors in the ACPC \cite{AAMAS2010,stitching,rgbr}. 
One of our aims in this paper is to help explain why CFR is performing well in non-zero-sum games. 

\subsection{Other Regret Minimization Concepts and Techniques}
\label{sec:related-work}

There are two other solution concepts associated with the notion of regret minimization. 
Both concepts define the regret of a strategy $\sigma_i$ to be
\[ regret_i(\sigma_i) = \max_{\substack{\sigma_i' \in \Sigma_i\\\sigma_{-i} \in \Sigma_{-i}}} u_i(\sigma_i', \sigma_{-i}) - u_i(\sigma_i, \sigma_{-i}). \]
Firstly, Renou and Schlag \cite{minimax-regret} define $\sigma^* \in \Sigma$ as a \emph{minimax regret equilibrium} relative to $\Sigma$ if
\[ regret_i(\sigma_i^*) \leq regret_i(\sigma_i) \text{ for all } \sigma_i \in \Sigma_i \text{ and all } i \in N. \]
This turns out to be an even stronger condition than Nash equilibrium, which is already hard to compute in games with more than two players. 
The authors also define the $\epsilon$\emph{-minimax regret equilibrium} variant where with probability $1 - \epsilon$ the opponents are assumed to play according to the equilibrium, and with probability $\epsilon$ no assumption is made. 
Here, the common assumption of rationality is dropped and thus $\epsilon$-minimax regret equilibria can end up playing iteratively strictly dominated strategies \cite[p.~276]{minimax-regret}. 

Secondly, Halpern and Pass \cite{IteratedRegretMin} introduce \emph{iterated regret minimization}. 
Much like iterated removal of dominated strategies, the authors iteratively remove all strategies $\sigma_i$ that do not provide minimal $regret_i(\sigma_i)$. 
They show that while the set of non-iteratively strictly dominated strategies can be disjoint from those that survive iterated regret minimization, their solutions match closely to those solutions played by real people in a number of small games. 
Our work here is less concerned with understanding how humans arrive at solutions and more concerned with understanding and advancing CFR in developing state-of-the-art game-playing agents.

\section{Dominated Actions}
\label{sec:dom-actions}

Our contributions in this paper begin with a formal definition of \emph{dominated actions} that are specific to extensive-form games, and we relate such actions to dominated strategies. 
Dominated actions are considered in the Gambit Software Tools package and are loosely defined as actions that are ``always worse than another, regardless of the beliefs at the information set'' \cite{Gambit}. 
Here, we say an action $a$ at $I \in \cI_i$ is a strictly dominated action if there exists a strategy $\sigma_i'$ that guarantees higher counterfactual value at $I$ to any other strategy $\sigma_i$ that always plays $a$ at $I$, regardless of what the opponents play but assuming they reach $I$ with positive probability. 
The formal definition is below. 
\begin{definition}
\label{def:domact}
An action $a \in A(I)$ of an extensive-form game is a \defword{strictly dominated action} if there exists a strategy $\sigma_i' \in \Sigma_i$ such that for all profiles $\sigma \in \Sigma$ satisfying $\sum_{h \in I} \pi_{-i}^\sigma(h) > 0$, we have $v_i(I, \sigma_{(I \rightarrow a)}) < v_i(I, (\sigma_i', \sigma_{-i}))$. 
\end{definition}
\noindent 
We use the counterfactual value $v_i$ instead of $u_i$ in Definition \ref{def:domact} because we are only concerned with the utility to player $i$ from $I$ onwards rather than over the entire game. 
Similar to iteratively dominated strategies, we also define an \defword{iteratively strictly dominated action} as one that is either strictly dominated or becomes strictly dominated after successively removing strictly dominated actions from the players' action sets. 
Analogous to strategic dominance in Definition \ref{def:domination}, \defword{weak} and \defword{very weak} action dominance allow equality rather than strict inequality for all but one profile $\sigma$ and for all profiles respectively. 
In addition, weak and very weak action dominance do not require the condition that $\sum_{h \in I} \pi_{-i}^\sigma(h) > 0$. 

For example, consider again Kuhn Poker defined in Section \ref{sec:background}. 
When player 2 is faced with a bet from player 1, calling the bet when holding the Jack is a strictly dominated action. 
This is because the Jack is the worst card and thus never wins regardless of player 1's private card. 
Similarly, folding with the King is a strictly dominated action. 
Note that a strategy that plays either of these actions with positive probability is not necessarily a strictly dominated strategy (but is a weakly dominated strategy, as Hoehn \etal~\cite{Hoehn2005} conclude) because there exist player 1 strategies that never bet. 
In addition, once these two actions are removed, one can check that player 1's action of betting with the Queen is iteratively strictly dominated. 
Since player 2 now only folds with the Jack and only calls with the King, it is strictly better for player 1 to always check with the Queen and then call a player 2 bet with probability $2/3$. 
Thus, iteratively strictly dominated actions can identify errors that iteratively strictly dominated strategies cannot.

Proposition \ref{prop:weakdom} below states a fundamental relationship between dominated actions and strategies. 
Any strategy that plays to reach information set $I$ ($\pi_i^\sigma(I) > 0$) and plays a weakly dominated action $a$ at $I$ ($\sigma_i(I,a) > 0$) is a weakly dominated strategy.
Since strictly dominated actions are also weakly dominated, it follows from Proposition \ref{prop:weakdom} that any strategy that plays a strictly dominated action is a weakly dominated strategy.
We provide a proof sketch of the proposition below, while full proofs can be found in \ref{sec:appendix}. 
\begin{prop}
\label{prop:weakdom}
If $a$ is a weakly dominated action at $I \in \cI_i$ and $\sigma_i \in \Sigma_i$ satisfies $\pi_i^\sigma(I)\sigma_i(I,a) > 0$, then $\sigma_i$ is a weakly dominated strategy. 
\end{prop}
\begin{proof-sketch} 
By definition of action dominance, there exists a strategy $\sigma_i' \in \Sigma_i$ such that $v_i(I, \sigma_{(I \rightarrow a)}) \leq v_i(I, (\sigma_i', \sigma_{-i})$ for all opponent profiles $\sigma_{-i} \in \Sigma_{-i}$. 
One can then construct a strategy $\sigma_i''$ that follows $\sigma_i$ everywhere except within the subtree rooted at $I$, where instead we follow a mixture of $\sigma_i$ and $\sigma_i'$. 
The weight in this mixture assigned to $\sigma_i'$ is $( 1 - \sigma_i(I,a) ) > 0$. 
The strategy $\sigma_i$ is then weakly dominated by $\sigma_i''$. \qed
\end{proof-sketch}

\begin{figure}[t]%
\centering
\includegraphics[width=0.98\textwidth]{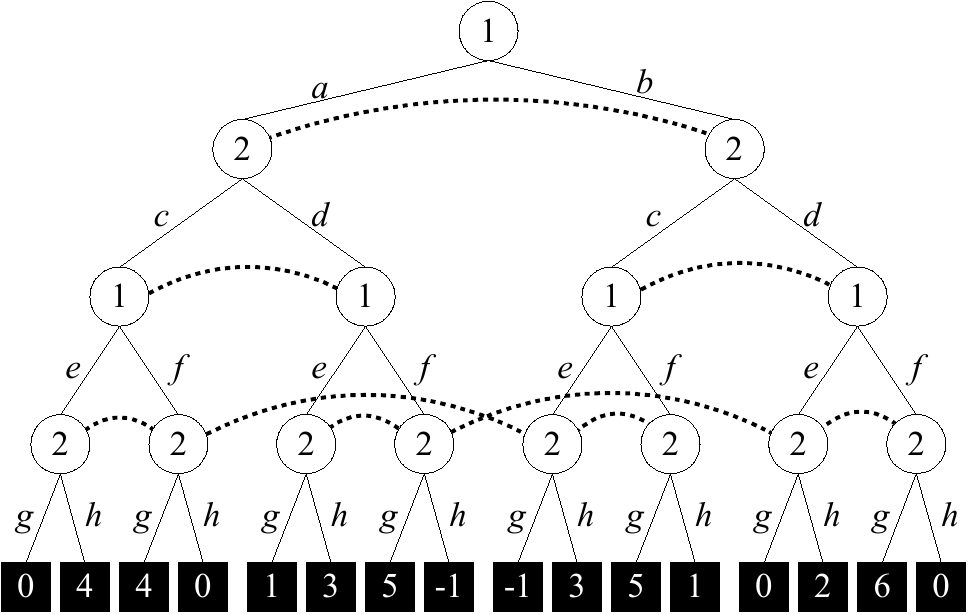}%
\caption{A zero-sum extensive-form game with strictly dominated strategies, but no strictly or weakly dominated actions. Nodes connected by a dashed line are in the same information set. Terminal values indicate utilities for player 1.}%
\label{fig:NoDomActs}%
\end{figure}

It is possible, however, for a dominated strategy to not play any dominated actions. 
For example, consider the zero-sum extensive-form game in Figure \ref{fig:NoDomActs} where both players take two private actions. 
The pure strategy for player 1 of playing $b$ and then $e$ is strictly dominated by the pure strategy that plays $a$ and then $e$ because the latter strategy guarantees exactly 1 more utility than the former, regardless of how player 2 plays. 
Similarly, the pure strategy that plays $a$ and then $f$ is strictly dominated by the pure strategy that plays $b$ and then $f$. 
However, no action is even weakly dominated. 
For instance, after playing $a$ (or $b$), the utility player 1 receives for playing $e$ can be greater, equal to, or less than the utility for playing $f$ depending on how player 2 plays. 

\section{Theoretical Analysis}
\label{sec:theory}

Clearly, one should never play a strictly dominated action or strategy as there always exists a better alternative. 
Furthermore, if we make the common assumption that our opponents are rational and do not play strictly dominated actions or strategies themselves, then we should never play iteratively strictly dominated actions or strategies. 
In zero-sum games, CFR converges to a Nash equilibrium, and so the average profile is guaranteed to eliminate strictly dominated strategies. 
For non-zero-sum games, however, Abou Risk and Szafron \cite{AAMAS2010} demonstrate that CFR may not converge to a Nash equilibrium. 
In this section, we provide theoretical evidence that CFR does eliminate (\ie, play with probability zero) strictly dominated actions and strategies. 

We begin by showing that in normal-form games, a class of regret minimization algorithms, including regret matching, all remove iteratively strictly dominated strategies. 
This is a simple result that, to our knowledge, was previously unknown. 
Recall that the support of a strategy $\sigma_i$, $\text{supp}(\sigma_i)$, is the set of actions assigned positive probability by $\sigma_i$. 
\begin{theorem}
\label{thm:normaldom}
Let $\sigma^1, \sigma^2, ...$ be a sequence of strategy profiles in a normal-form game where all players' strategies are computed by regret minimization algorithms where for all $i \in N$, $a \in A_i$, if $R_i^T(a) < 0$ and $R_i^T(a) < \max_{b \in A_i} R_i^T(b)$, then $\sigma_i^{T+1}(a) = 0$. 
If $\sigma_i$ is an iteratively strictly dominated strategy, then there exists an integer $T_0$ such that for all $T \geq T_0$, $\text{\emph{supp}}(\sigma_i) \nsubseteq \text{\emph{supp}}(\sigma_i^T)$. 
\end{theorem}
\begin{proof-sketch} 
For the non-iterative dominance case, by strict domination of $\sigma_i$, there exists another strategy $\sigma_i' \in \Sigma_i$ such that 
\[ \epsilon = \min_{a_{-i} \in A_{-i}} u_i( \sigma_i', a_{-i} ) - u_i( \sigma_i, a_{-i} ) > 0. \]
One can then show that there exists an action $a \in \text{supp}(\sigma_i)$ such that
\[ R_i^T(a) \leq -\epsilon T + \max_{b \in A_i} R_i^T(b) \leq -\epsilon T + R_i^{T,+}. \]
Since $R_i^{T,+} / T \rightarrow 0$ as $T \rightarrow \infty$, it follows that $R_i^T(a) < 0$ after some finite number of iterations $T_0$. 
By our assumption, this implies $a \notin \text{supp}(\sigma_i^T)$ for all $T \geq T_0$ as desired. 
Using the fact that new iterative dominances only arise from removing actions and never from removing mixed strategies \cite{DominanceComplexity}, iterative dominance is proven by induction on the finite number of iteratively dominated pure strategies that must first be removed to exhibit domination of $\sigma_i$. \qed
\end{proof-sketch}

Note that regret matching is a regret minimization algorithm that satisfies the conditions required by Theorem \ref{thm:normaldom}, as long as when the denominator of equation \eqref{eq:RM} is zero, we choose $\sigma_i^{T+1}(a) = 0$ when $R_i^T(a) < \max_{b \in A_i}R_i^T(b)$. 
Also, if a pure strategy $s_i(a) = 1$ is iteratively strictly dominated, then Theorem \ref{thm:normaldom} implies that $\sigma_i^T$ never plays action $a$ after a finite number of iterations. 

We now turn our attention to extensive-form games, which are our primary concern. 
Here, the linear program \eqref{lp:dom} cannot be applied to find non-iteratively strictly dominated strategies in even moderately-sized extensive-form games as the programs would require a number of constraints exponential in the size of the game. 
On the other hand, we can apply CFR. 

First, we consider the removal of iteratively strictly dominated actions. 
Our results rely on two conditions. 
Let $x^T$ be the number of iterations $t$ where $\sum_{a \in A(I)} R_i^{t,+}(I,a) = 0$ for some $i \in N$ and $I \in \cI_i$, $1 \leq t \leq T$. 
The first condition we require is that $x^T$ be sublinear in $T$. 
Intuitively, this is necessary because otherwise, the denominator of equation \eqref{eq:CFRRM} is zero too often, and so regret matching too often yields an arbitrary strategy at some $I \in \cI_i$ that potentially plays a dominated action. 
While we cannot prove that this condition always holds, we show empirically that $x^T / T$ decreases over time in the next section. 
Next, for $I \in \cI_i$ and $\delta \geq 0$, define $\Sigma_\delta(I) = \{ \sigma \in \Sigma \mid \sum_{h \in I} \pi_{-i}^\sigma(h) \geq \delta \}$ to be the set of profiles where the probability that the opponents play to reach $I$, $\sum_{h \in I} \pi_{-i}^\sigma(h)$, is at least $\delta$. 
The second condition we require is that the opponents reach each information set $I$ containing a dominated action \emph{often enough}, meaning that there exist real numbers $\delta, \gamma > 0$ and an integer $T'$ such that for all $T \geq T'$, $| \Sigma_{\delta}(I) \cap \{ \sigma^t \mid T' \leq t \leq T \} | \geq \gamma T$. 
This condition appears necessary because the magnitude of the counterfactual regret $|r_i^t(I,a)| = |v_i(I, \sigma_{(I \rightarrow a)}^t) - v_i(\sigma^t)| \leq \Delta_i\sum_{h \in I} \pi_{-i}^{\sigma^t}(h)$ is weighted by the probability of the opponents reaching $I$. 
Thus, if the opponents reach $I$ with probability zero, then we will stop \emph{learning} how to adjust our strategy. 
Since it could take several iterations to eliminate an iteratively strictly dominated action, we may end up stuck playing such an action when $I$ is not reached by the opponents often enough. \begin{theorem}
\label{thm:domacts}
Let $\sigma^1, \sigma^2, ...$ be strategy profiles generated by CFR in an extensive-form game, let $I \in \cI_i$, and let $a$ be an iteratively strictly dominated action at $I$, where removal in sequence of the iteratively strictly dominated actions $a_1, ..., a_k$ at $I_1, ..., I_k$ respectively yields iterative dominance of $a_{k+1} = a$. 
If for $1 \leq \ell \leq k+1$, there exist real numbers $\delta_\ell, \gamma_\ell > 0$ and an integer $T_\ell$ such that for all $T \geq T_\ell$, $| \Sigma_{\delta_\ell}(I_\ell) \cap \{ \sigma^t \mid T_\ell \leq t \leq T \} | \geq \gamma_\ell T$, then
\begin{itemize}
\item[\emph{(i)}] there exists an integer $T_0$ such that for all $T \geq T_0$, $R_i^T(I,a) < 0$,
\item[\emph{(ii)}] if $\lim_{T \rightarrow \infty} x^T / T = 0$, then $\lim_{T \rightarrow \infty} y^T(I,a) / T = 0$, where $y^T(I,a)$ is the number of iterations $1 \leq t \leq T$ satisfying $\sigma^t(I,a) > 0$, and
\item[\emph{(iii)}] if $\lim_{T \rightarrow \infty} x^T / T = 0$, then $\lim_{T \rightarrow \infty} \pi_i^{\bar{\sigma}^T}(I) \bar{\sigma}_i^T(I,a) = 0$. 
\end{itemize}
\end{theorem}
\begin{proof-sketch} 
Similar to the proof of Theorem \ref{thm:normaldom}, there exists an $\epsilon > 0$ and a term $F$ such that 
\[ R_i^T(I,a) \leq -\epsilon \gamma T + F \text{ where } \lim_{T \rightarrow \infty} \frac{F}{T} = 0. \] 
Again, this implies that there exists an integer $T_0$ such that for all $T \geq T_0$, $R_i^T(I,a) < 0$, establishing part (i). 
Since CFR applies regret matching at $I$, part (i) and equation \eqref{eq:RM} imply that for all $T \geq T_0$, either $\sum_{b \in A(I)} R_i^{T,+}(I,b) = 0$ or $\sigma_i^{T+1}(I,a) = 0$. 
From this, we have
\[ \lim_{T \rightarrow \infty} \frac{y^T(I,a)}{T} \leq \lim_{T \rightarrow \infty} \frac{y^{T_0}(I,a) + x^T}{T} = 0, \]
proving part (ii). 
Finally, part (iii) follows according to
\[ \lim_{T \rightarrow \infty} \pi_i^{\bar{\sigma}^T}(I) \bar{\sigma}_i^T(I,a) = \lim_{T \rightarrow \infty} \frac{\sum_{t=1}^T \pi_i^{\sigma^t}(I) \sigma_i^t(I,a)}{T} \leq \lim_{T \rightarrow \infty} \frac{y^T(I,a)}{T} = 0,\]
where the first equality is by the definition of the average strategy and the inequality is by definition of $y^T(I,a)$. \qed
\end{proof-sketch}

Part (iii) of Theorem \ref{thm:domacts} says that an iteratively strictly dominated action is not reached or is removed from the average profile $\bar{\sigma}^T$ in the limit, whereas part (i) suggests that iteratively strictly dominated actions are removed from the \emph{current} profile $\sigma^T$ after just a finite number of iterations (except possibly when $\sum_{a \in A(I)} R_i^{T,+}(I,a) = 0$). 
Finally, part (ii) states that the number of current profiles that play an iteratively strictly dominated action $a$ at $I$, $y^T(I,a)$, is sublinear in $T$. 

Next, we show that the profiles generated by CFR eliminate all iteratively strictly dominated strategies, assuming again that $x^T / T \rightarrow 0$. 
\begin{theorem}
\label{thm:domstrats}
Let $\sigma^1, \sigma^2, ...$ be strategy profiles generated by CFR in an extensive-form game, and let $\sigma_i$ be an iteratively strictly dominated strategy. Then,
\begin{itemize}
\item[\emph{(i)}] there exists an integer $T_0$ such that for all $T \geq T_0$, there exist $I \in \cI_i$, $a \in A(I)$ such that $\pi_i^\sigma(I)\sigma_i(I,a) > 0$ and $R_i^T(I,a) < 0$, and
\item[\emph{(ii)}] if $\lim_{T \rightarrow \infty} x^T / T = 0$, then $\lim_{T \rightarrow \infty} y^T(\sigma_i) / T = 0$, where $y^T(\sigma_i)$ is the number of iterations $1 \leq t \leq T$ satisfying $\text{supp}(\sigma_i) \subseteq \text{supp}(\sigma_i^t)$.
\end{itemize} 
\end{theorem}
\begin{proof-sketch} 
For $\sigma_i \in \Sigma_i$, define 
\[ R_{i,\text{full}}^T(\sigma_i) = \sum_{t=1}^T ( u_i(\sigma_i, \sigma_{-i}^t) - u_i(\sigma^t) ). \]
Similar to the proof of Theorems \ref{thm:normaldom} and \ref{thm:domacts}, there exists an $\epsilon > 0$ and a term $F'$ such that 
\begin{equation}
\label{eq:fulllessthanzero}
R_{i,\text{full}}^T(\sigma_i) \leq -\epsilon T + F' \text{ where } \lim_{T \rightarrow \infty} \frac{F'}{T} = 0.
\end{equation}
Next, one can show that 
\begin{equation}
\label{eq:full}
R_{i,\text{full}}^T(\sigma_i) = \sum_{I \in \cI_i} \pi_i^\sigma(I) \sum_{a \in A(I)} \sigma_i(I,a) R_i^T(I,a). 
\end{equation}
Since $\pi_i^\sigma(I), \sigma_i(I,a) \geq 0$, it follows by equations \eqref{eq:fulllessthanzero} and \eqref{eq:full} that after a finite number of iterations $T_0$, there exist $I \in \cI_i$, $a \in A(I)$ such that $\pi_i^\sigma(I) \sigma_i(I,a) > 0$ and $R_i^T(I,a) < 0$, establishing part (i). 
Part (ii) then follows as in the proof of part (ii) of Theorem \ref{thm:domacts}. \qed
\end{proof-sketch}

\begin{figure}%
\centering
\includegraphics[width=0.6\columnwidth]{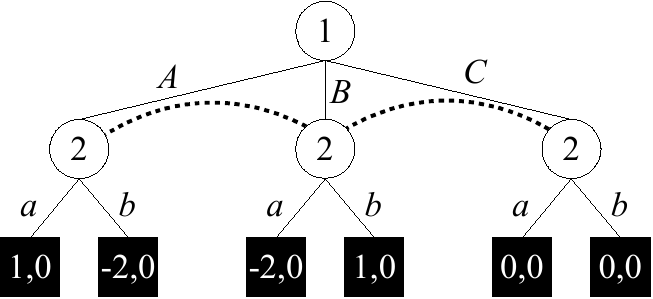}%
\caption{A two-player non-zero-sum extensive-form game where each player has a single information set.}%
\label{fig:domavg}%
\end{figure}

Similar to part (i) of Theorem \ref{thm:domacts}, part (i) of Theorem \ref{thm:domstrats} says that after a finite number of iterations, there is always some information set $I$ that the dominated strategy $\sigma_i$ plays to reach and some action at $I$ played by $\sigma_i$ which $\sigma_i^T$ does not play (except possibly when $\sum_{a \in A(I)} R_i^{T,+}(I,a) = 0$), and so $\sigma_i^T \neq \sigma_i$.  
Part (ii) similarly states that the number of profiles generated whose support contains $\sigma_i$, $y^T(\sigma_i)$, is sublinear in $T$. 
Notice that Theorems \ref{thm:normaldom} and \ref{thm:domstrats} do not draw any conclusions upon the average profile $\bar{\sigma}^T$. 
Perhaps surprisingly, it is possible to have a sequence of profiles with no regret where the average profile converges to a strictly dominated strategy. 
Consider the two-player non-zero-sum game in Figure \ref{fig:domavg}. 
The sequence of pure strategy profiles $(A,a), (B,b), (A,a), (B,b), ...$ has no positive regret for either player, and in the limit, the average profile for player 1, $\bar{\sigma}_1^T$, plays $A$ and $B$ each with probability $0.5$. 
However, $\bar{\sigma}_1^T$ is strictly dominated by the pure strategy that always plays $C$. 

Our final theoretical contribution shows that in two-player non-zero-sum games, regret minimization yields a bound on the average strategy profile's distance from being a Nash equilibrium. 
\begin{theorem}
\label{thm:genfolk}
Let $\epsilon, \delta \geq 0$ and let $\sigma^1, \sigma^2, ..., \sigma^T$ be strategy profiles in a two-player game. 
If $R_i^T / T \leq \epsilon$ for $i=1,2$, and $|u_1 + u_2| \leq \delta$, then $\bar{\sigma}^T$ is a $2(\epsilon + \delta)$-Nash equilibrium. 
\end{theorem}
\begin{proof}
We generalize the proof of \cite[p.~11]{Waugh2009}. 
For $i = 1,2$, by the definition of regret, we have
\begin{align*}
\epsilon &\geq \frac{1}{T} \max_{\sigma_i' \in \Sigma_i} \sum_{t=1}^T \left( u_i(\sigma_i', \sigma_{-i}^t) - u_i(\sigma^t) \right) \\
&= \max_{\sigma_i' \in \Sigma_i} u_i(\sigma_i', \bar{\sigma}_{-i}^T) - \frac{1}{T} \sum_{t=1}^T u_i(\sigma^t)
\end{align*}
by linearity of expectation. 
Summing the two inequalities for $i=1,2$ gives
\begin{align*}
2\epsilon &\geq \max_{\sigma_1' \in \Sigma_1} u_1(\sigma_1', \bar{\sigma}_2^T) + \max_{\sigma_2' \in \Sigma_2} u_2(\bar{\sigma}_1^T, \sigma_2') - \frac{1}{T} \sum_{t=1}^T \left( u_1(\sigma^t) + u_2(\sigma^t) \right) \\
&\geq \max_{\sigma_1' \in \Sigma_1} u_1(\sigma_1', \bar{\sigma}_2^T) + \max_{\sigma_2' \in \Sigma_2} \left( -u_1(\bar{\sigma}_1^T, \sigma_2') - \delta \right) - \delta \\
&= \max_{\sigma_1' \in \Sigma_1} u_1(\sigma_1', \bar{\sigma}_2^T) - \min_{\sigma_2' \in \Sigma_2} u_1(\bar{\sigma}_1^T, \sigma_2') - 2\delta \\
&\geq \max_{\sigma_1' \in \Sigma_1} u_1(\sigma_1', \bar{\sigma}_2^T) - u_1(\bar{\sigma}^T) - 2\delta,
\end{align*}
where the last line follows by setting $\sigma_2' = \bar{\sigma}_2^T$. 
Rearranging terms gives 
\[ \max_{\sigma_1' \in \Sigma_1} u_1(\sigma_1', \bar{\sigma}_2^T) \leq u_1(\bar{\sigma}^T) + 2(\epsilon + \delta). \]
Applying the same arguments but reversing the roles of the two players gives 
\[ \max_{\sigma_2' \in \Sigma_2} u_2(\bar{\sigma}_1^T, \sigma_2') \leq u_2(\bar{\sigma}^T) + 2(\epsilon + \delta), \]
and thus by definition $\bar{\sigma}^T$ is a $2(\epsilon + \delta)$-Nash equilibrium. 
\end{proof}

Theorem \ref{thm:genfolk} is a generalization of Theorem \ref{thm:folk}. 
When $\delta = 0$, the game is zero-sum, and so the average profile converges to equilibrium as $\epsilon \rightarrow 0$. 
In addition, when the players' utilities sum to at most $\delta > 0$, then as $\epsilon \rightarrow 0$, the average profile converges to a $2\delta$-Nash equilibrium.

\subsection{Remarks} 

Theorems \ref{thm:normaldom}, \ref{thm:domacts}, and \ref{thm:domstrats} provide evidence that regret minimization removes iterative strict domination. 
Of course, eliminating strict domination may not provide any useful insights in games where few strategies are iteratively strictly dominated. 
Despite this obvious limitation, Theorems \ref{thm:domacts} and \ref{thm:domstrats} provide a better understanding of the strategies generated by CFR in non-zero-sum games than what coarse correlated equilibria provide. 
In the next section, we show that avoiding iteratively dominated actions is enough to perform well in Kuhn Poker. 
However, large games such as three-player Texas Hold'em are too complex to analyze action and strategic dominance beyond obvious errors, such as folding the best hand. 
It remains open as to how well our theory explains the success of CFR in these large games. 

Perhaps more importantly, the theory developed here has guided us to a more efficient adaptation of CFR, in both time and memory, for games with more than two players. 
Given Theorems \ref{thm:domacts} and \ref{thm:domstrats} and given we have only finite time, we suggest using the current profile in practice rather than the average. 
In fact, while Theorem \ref{thm:genfolk} says that the average profile converges to a $2\delta$-Nash equilibrium in two-player games, there is no clear case for preferring the average over the current profile in three-or-more-player games. 
Furthermore, the average profile is not used in any computations during CFR, so when discarding the average, there is no reason to store the cumulative profile. 
This reduces the memory requirements of CFR by a factor of two, since then only one value per information set, action pair ($R_i^T(I,a)$) must be stored as opposed to two. 
Not only does this allow us to tackle larger games, the extra memory might be utilized to compute even stronger strategies than previously possible. 

We are not the first to consider using the current profile. 
In CFR-BR, a recently developed CFR variant for zero-sum games that replaces one player with a worst-case opponent, the current profile converges to equilibrium with high probability \cite[Theorem 4]{cfr-br}. 
The authors discuss similar benefits when discarding the cumulative profile in CFR-BR and just using the current strategy profile. 
Nonetheless, we are the first to suggest using the current profile both with the original CFR algorithm and in games with more than two players. 
The next section explores these new insights. 

\section{Empirical Study}
\label{sec:experiments}

Using poker as a testbed, we design several experiments to test our theory developed in the previous section. 
While previous work has applied CFR across several domains \cite{lanctot-thesis}, poker games are of particular interest as they are widely popular and many computer agents from past ACPC events are available to test against. 
New games can also be easily created by adjusting the number of players, cards, and betting rounds that take place. 

\subsection{Poker Games}
\label{sec:poker-games}


We consider three different poker games for our experiments in this section. 
The first is Kuhn Poker, which was introduced in Section \ref{sec:background}. 

Our second game and our main game of interest is three-player limit Texas Hold'em. 
To begin the game, the player to the left of the dealer posts a small blind of five chips, the next player posts a big blind of ten chips, and each player is dealt two private cards from a standard 52-card deck. 
Texas Hold'em consists of four betting rounds with three, one, and one public card(s) being revealed before the second, third, and fourth rounds respectively. 
All bets and raises are fixed to ten chips in the first two rounds and twenty chips in the last two rounds; players may not go \emph{all-in} as in no-limit poker. 
There is also a maximum of four bets or raises allowed per round. 
At the end of the fourth round, the players that did not fold reveal their hand. 
The player with the highest ranked poker hand made up of any combination of their two private cards and five public cards wins all the chips played. 

With three players, limit Texas Hold'em contains approximately $5 \times 10^{17}$ information sets and CFR would require hundreds of petabytes of RAM to minimize regret in such a large game. 
Instead, a common approach is to use state-space abstraction to produce a similar game of a tractable size by merging information sets or restricting the action space~\cite{Gilpin&Sandholm2006,Abstraction2013}. 
For Texas Hold'em, we merge card deals into \emph{buckets} so that hands falling into the same bucket are indistinguishable. 
We can then control the size of the abstract game by increasing or decreasing the number of buckets used on each round. 
However, increasing abstraction size not only increases memory requirements, but also increases the number of iterations required to minimize average regret (see equation \eqref{eq:CFRBound}). 
There are just three actions (fold, check/call, and bet/raise) available in limit Hold'em, and thus we do not abstract on actions. 
Note that applying CFR to an abstraction of Texas Hold'em yields no guarantees about regret minimization or domination avoidance in the real game (but are guaranteed in the abstract game). 
Furthermore, we will use imperfect recall abstractions that forget the buckets from previous rounds and break our assumption of perfect recall stated in Section \ref{sec:background}. 
Despite these complications, abstraction and imperfect recall still appear to work well in practice \cite{ImperfectRecall2009,rgbr}. 

Thirdly, we also consider the game of \emph{2-1 Hold'em} \cite{pcs} that is identical to Texas Hold'em, except consists of only the first two betting rounds and only one raise is allowed per round.  
Two-player 2-1 Hold'em has roughly 16 million information sets, which is small enough to apply CFR without abstraction. 
Furthermore, because full tree traversals in CFR are very expensive, we instead use sampling variants that only traverse a smaller subset of information sets on each iteration. 
We found that the most efficient variant for 2-1 Hold'em was Public Chance Sampling \cite{pcs} and for three-player limit Texas Hold'em was External Sampling \cite{mccfr}. 


\subsection{Dominated Actions and Performance in Kuhn Poker}
\label{sec:kuhn-experiments}

\begin{table}[t]%
\centering
\caption{Results of a six-agent mock tournament of Kuhn poker. Reported scores for the row strategy against the column strategy are in expected milli-chips per game, averaged over both player orderings. \label{tab:kuhn-tourney}}{
\begin{tabular}{r|ccccccc}
 & \bf Uni & \bf ND & \bf NID & \bf NE-0 & \bf NE-0.5 & \bf NE-1 & \bf Overall \\
\hline
\bf Uni &  - & -270 & -187 & -111 & -138 & -166 & -174 \\
\bf ND & 270 & - & -31 & -55 & -34 & -13 & 27 \\
\bf NID & 187 & 31 & - & 0 & 0 & 0 & 43 \\
\bf NE-0 & 111 & 55 & 0 & - & 0 & 0 & 33 \\
\bf NE-0.5 & 138 & 34 & 0 & 0 & - & 0 & 34 \\
\bf NE-1 & 166 & 13 & 0 & 0 & 0 & - & 36 \\
\end{tabular}
}
\end{table}

To begin, we investigate the correlation between the presence of iteratively dominated actions in one's strategy with the performance of the strategy in a mock ACPC-style tournament. 
In the ACPC, each game is evaluated according to two different scoring metrics. 
The total bankroll (TBR) metric simply ranks competitors according to their overall earnings in money per game averaged across all possible opponents. 
The instant runoff (IRO) metric, however, ranks competitors by iteratively eliminating the lowest scoring agent from consideration and reevaluating the overall scores by averaging only across the remaining agents. 
In a zero-sum game, a Nash equilibrium strategy is optimal for winning IRO since it never loses in expectation to any opponent. 

We ran a six-agent mock tournament of Kuhn poker, which was introduced in Section \ref{sec:background}. 
Kuhn poker is a small enough game where we can easily identify all iteratively dominated actions and all Nash equilibrium strategies have already been classified \cite{Kuhn}. 
Our agents consist of a uniform random strategy (Uni), a strategy that plays no dominated actions (does not call with the Jack or fold with the King) but is otherwise uniform random (ND), a strategy that plays no iteratively dominated actions (no dominated actions and does not bet with the Queen) but is otherwise uniform random (NID), and three Nash equilibrium strategies (NE-$\gamma$) for $\gamma = 0, 0.5, 1$, where $\gamma$ is the probability of betting with the King. 
A cross table of the results for each pair of strategies is given in Table \ref{tab:kuhn-tourney}. 
For IRO, after successively eliminating Uni and then ND, there is a four-way tie for first place between the three equilibrium strategies and NID. 
In addition, NID happens to win TBR, though none of the strategies are designed with TBR in mind. 
This mock tournament provides one example where high performance can be achieved by simply avoiding iteratively dominated errors. 

\subsection{Distance from Equilibrium in Two-Player Non-Zero-Sum Games}

\begin{figure}[t]%
\centering
\subfloat[Orange tilt] {
	\includegraphics[width=0.48\columnwidth]{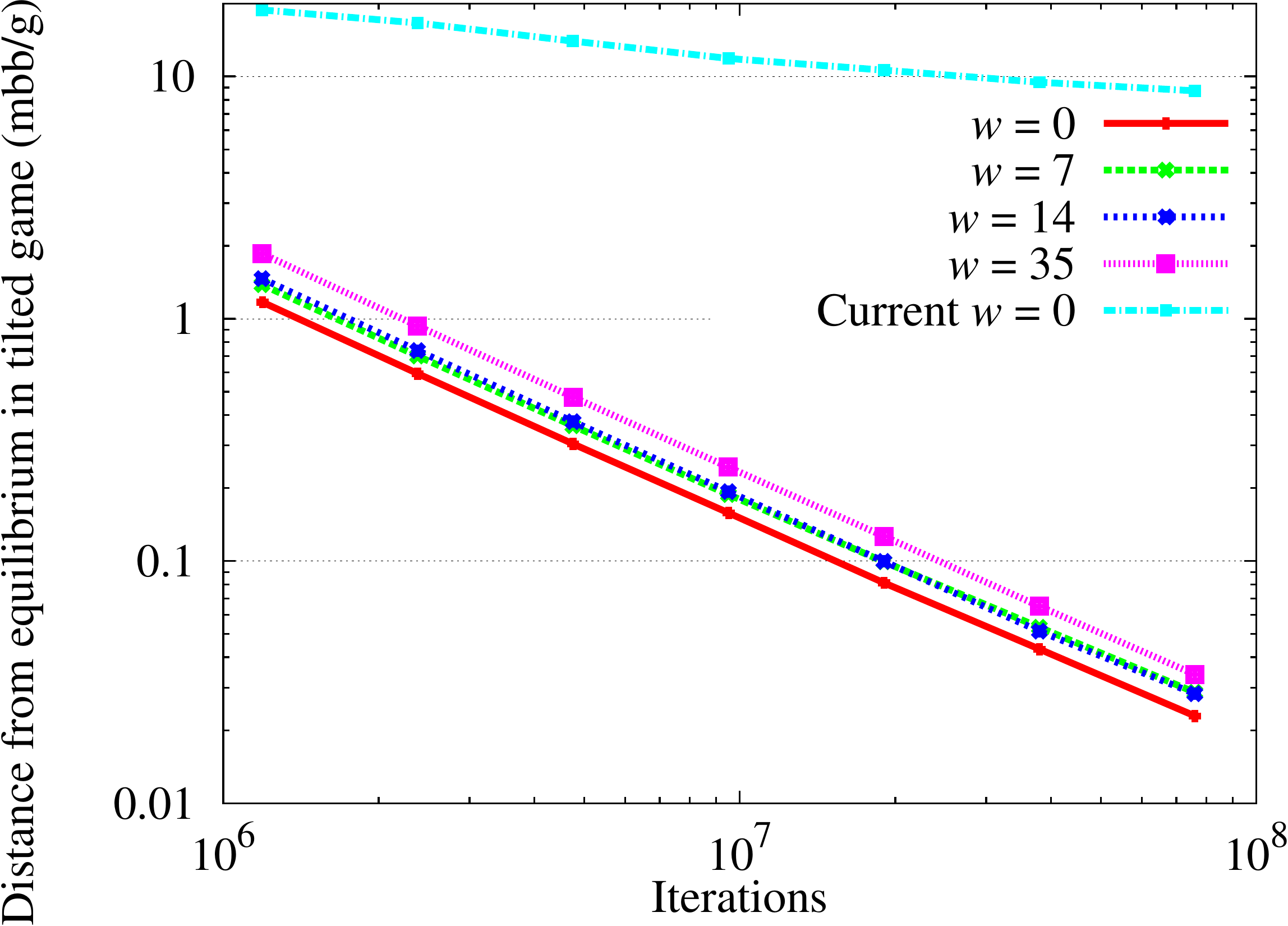}%
	\label{fig:orange}
} 
\subfloat[Green tilt] {
	\includegraphics[width=0.48\columnwidth]{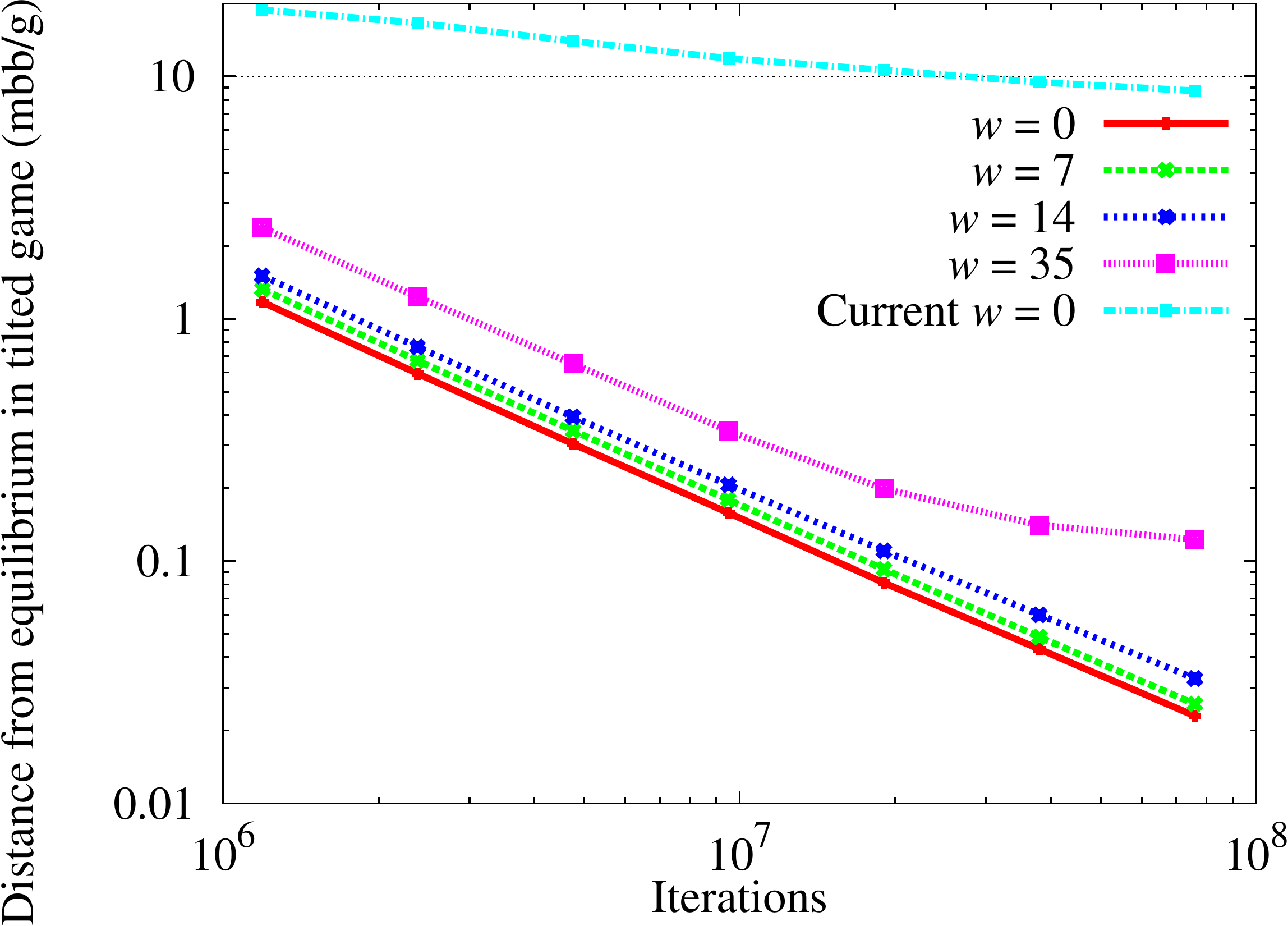}%
	\label{fig:green}
}
\caption{Log-log plots measuring the distance from equilibrium of CFR strategies in \textbf{\boldmath{$w$}}\%-tilted 2-1 Hold'em over iterations.  Distance is measured in milli-big-blinds per game (mbb/g).}%
\label{fig:tilts}%
\end{figure}

Our next experiment applies CFR to non-zero-sum \emph{tilted} variants of two-player 2-1 Hold'em. 
Tilted games are constructed by rewarding or penalizing players depending on the outcome of the game. 
This can lead to more aggressive play when applied to the regular, non-tilted game and were used by the poker program Polaris that won the 2008 Man-vs-Machine competition \cite{rgbr}. 
Here, we use the \emph{orange} tilt that gives the winning player an extra $w$\% bonus, and the \emph{green} tilt that both reduces the losing player's loss in a showdown (\ie,  when neither player folded) by $w$\% and penalizes the winning player by $w$\% when the losing player folded. 
In both of these games, we can bound $|u_1 + u_2| \leq \Delta_i w / 100$, and so Theorem \ref{thm:genfolk} states that CFR will converge to at least a $\Delta_i w / 50$-Nash equilibrium. 
For $w \in \{0,7,14,35\}$, we ran CFR and measured how far the average profile was from equilibrium in the $w$\%-tilted game by calculating $\max_{\sigma_i \in \Sigma_i} u_i(\sigma_i, \bar{\sigma}_{-i}^T)$ and averaging over both players $i=1,2$. 
In addition, we also measured the same value for the current profile in the non-tilted game ($w = 0$). 
These results are shown in Figure \ref{fig:tilts}. 
As expected, in the non-tilted game ($w = 0$), the average profile is approaching a Nash equilibrium. 
For the tilted games, we see that as $w$ is increased, most of the profiles are further from equilibrium, coinciding with Theorem \ref{thm:genfolk}. 
However, the strategies are much closer to equilibrium than the distance guaranteed by Theorem \ref{thm:genfolk} (note that $\Delta_i = 8$ big blinds) and only in the green tilt with $w=35$ is it obvious that CFR is not converging to an exact equilibrium. 
Of course, Theorem \ref{thm:genfolk} only provides an upper bound on the average profile's distance from equilibrium, and this bound appears to be quite loose. 
These results warrant further investigation into regret minimization in two-player non-zero-sum games. 
Finally, it is clear that the current strategy profile with $w=0$ is not converging to equilibrium. 
Thus, unlike CFR-BR \cite{cfr-br}, the average profile from CFR is generally preferred to the current profile in two-player games as it gives a better worst-case guarantee. 

\subsection{Positive Regret and Current Profile in Three-Player Limit Hold'em}
\label{sec:holdem-experiments}

\begin{figure}[t]%
\centering
\subfloat[] {
	\includegraphics[width=0.49\columnwidth]{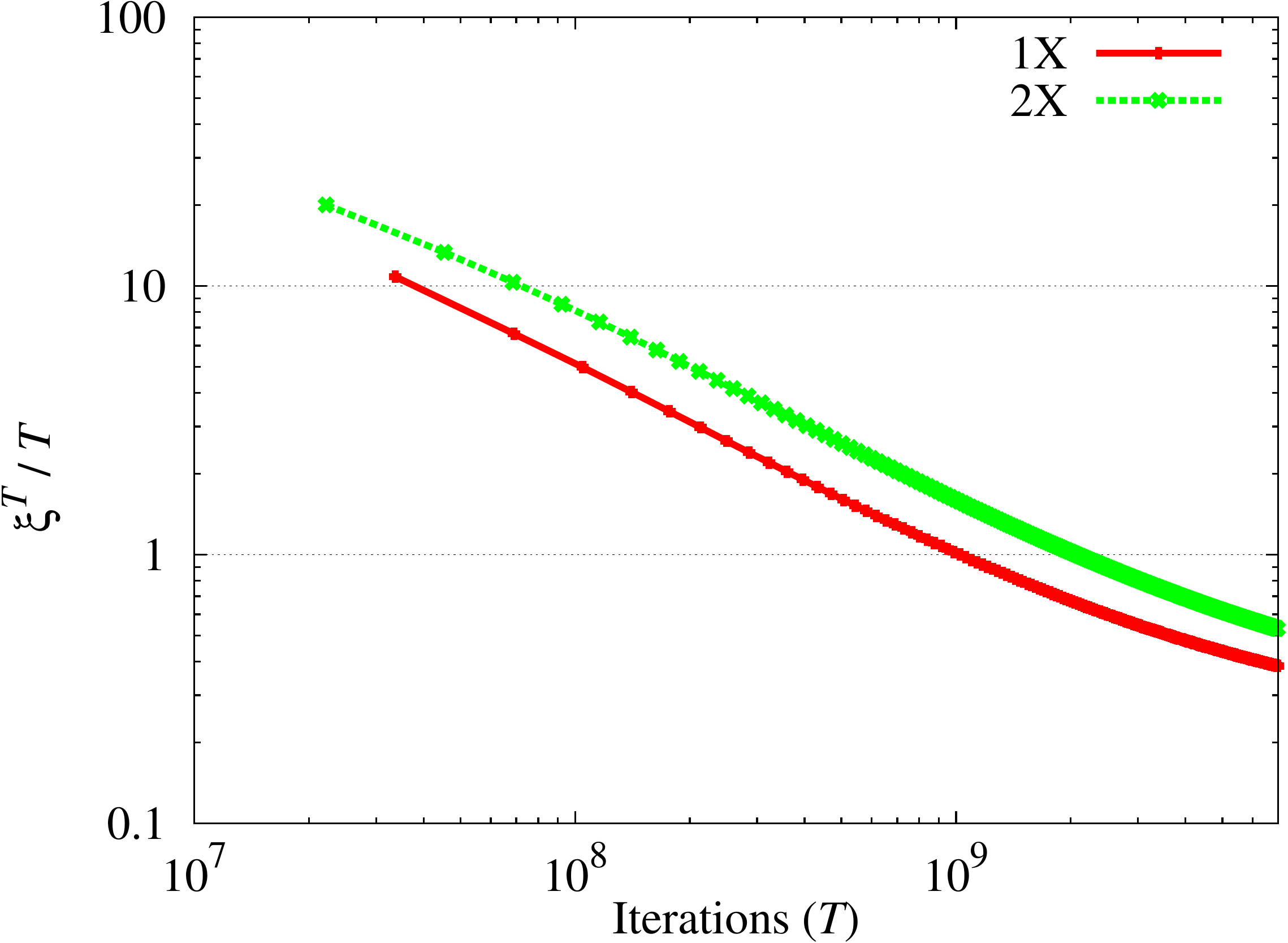}%
	\label{fig:posreg}%
}
\subfloat[] {
	\includegraphics[width=0.49\columnwidth]{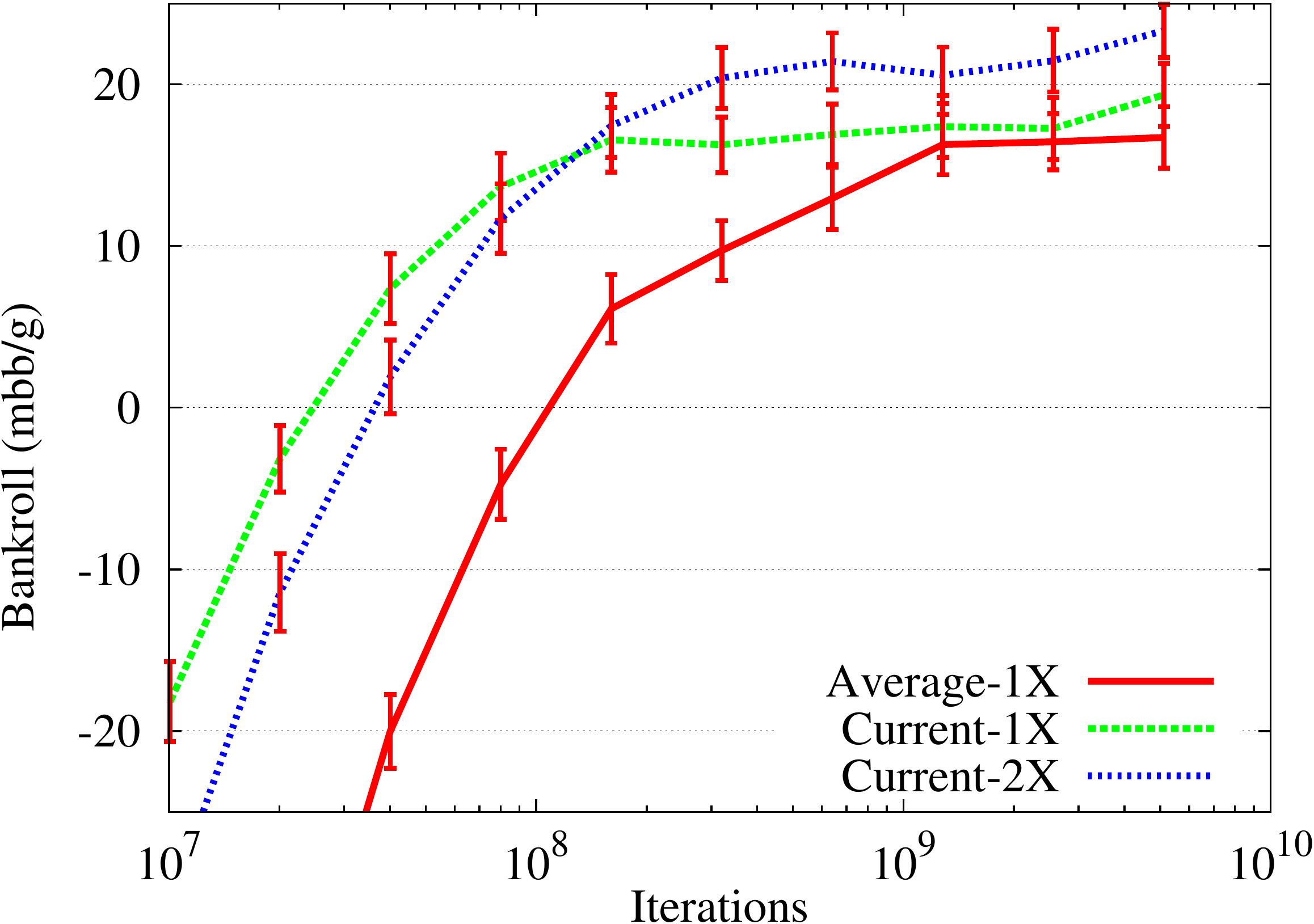}%
	\label{fig:2Xabs}%
}
\caption{\textbf{(a)} Log-log plot measuring the frequency at which an information set is visited where every action has nonpositive cumulative counterfactual regret during CFR in the 1X and 2X abstractions of three-player limit Texas Hold'em.  \textbf{(b)} Performance over iterations (log scale) of three strategy profiles in a four-agent round-robin competition of three-player limit Texas Hold'em, measured in milli-big-blinds per game. \emph{Current-2X} is the current profile generated by CFR in the 2X abstraction that is twice as large as the 1X abstraction used to generate \emph{Average-1X} and \emph{Current-1X}.  Error bars indicate 95\% confidence intervals over 50 competitions.}%
\end{figure}

Next, we examine how often $\sum_{a \in A(I)} R_i^{T,+}(I,a) = 0$ as required by parts of Theorems \ref{thm:domacts} and \ref{thm:domstrats}. 
CFR was applied to two different abstractions of three-player limit Texas Hold'em. 
The first, labeled \emph{1X}, consists of 169, 900, 100, and 25 buckets per betting round respectively. 
This abstraction size was used by the winning agents of the 2010 ACPC 3-player events \cite{stitching} and contains about 262 million information sets. 
The second abstraction, labeled \emph{2X}, uses 169, 1800, 200, and 50 buckets per betting round respectively, resulting in an abstract game approximately twice the size. 
All of our abstractions were built off-line using a $k$-means clustering algorithm on hand strength distribution described by Johanson \etal~\cite{Abstraction2013}. 
For each abstraction, we measured $\xi^T$, 
the total number of times where External Sampling traversed an information set that had no positive regret at any action. 
The average of $\xi^T$ is plotted over iterations $T$ in Figure \ref{fig:posreg}. 
In both cases, we see that encountering an information set with no positive regret becomes less frequent over time, where we eventually encounter fewer than one such information set per iteration on average. 
While we cannot guarantee that $x^T / T \approx \xi^T / T \rightarrow 0$ as required by Theorems \ref{thm:domacts} and \ref{thm:domstrats}, we at least have evidence that having no positive regret becomes a rare event. 
By part (i) of Theorems \ref{thm:domacts} and \ref{thm:domstrats}, this means that iteratively strictly dominated actions and strategies will likely be avoided in the current strategy profile. 

Using these same abstractions of three-player Hold'em, we now show that the current profile can reach higher performance faster than the average profile, and that the extra savings in memory acquired by discarding the average profile can be utilized to generate even stronger strategies. 
In this experiment, we generated three different strategy profiles with CFR, saving the profiles at various iteration counts.  
For the 1X abstraction, we kept both the average and the current profile, while for the 2X abstraction, we kept just the current profile. 
Note that running CFR on the 2X abstract game without keeping the average profile requires no more RAM than running CFR on the 1X abstraction and keeping both profiles. 
For each of our saved profiles, we then played a four-agent round-robin competition (RRC) against the base strategy profiles\footnote{The 2010 and 2011 agents employed special experts in some two-player scenarios that were not used in this specific experiment. More details regarding these agents are provided later in Section \ref{sec:acpc}.} from the top 2009, 2010, and 2011 ACPC three-player entries. 
Figure \ref{fig:2Xabs} shows the amount won by each of our three strategies over iterations, averaged over 50 RRCs consisting of 10,000 games per match. 
Clearly, the 1X current profile reaches strong play much sooner than the average profile, which requires about ten times the number of iterations to reach the same level of performance. 
Furthermore, while more iterations are needed in the 2X abstraction as expected by equation \eqref{eq:CFRBound}, we see that 2X eventually yields a current profile that outperforms both profiles in the 1X abstraction. 

\section{A New Champion Three-Player Limit Texas Hold'em Agent}
\label{sec:acpc}

Finally, this section presents a new three-player limit Texas Hold'em agent that won the three-player events of the 2012 ACPC. 
Before presenting this new agent in detail, we summarize the previous competition winners. 

\subsection{Previous ACPC Winners}
\label{sec:old-acpc}

As we discussed in Section \ref{sec:poker-games}, abstraction is necessary in order to feasibly apply CFR to Texas Hold'em. 
Despite the loss of theoretical guarantees and the existence of abstraction pathologies \cite{AbstractionPathologies}, we generally see increased performance as we increase the granularity of our abstractions; in other words, more buckets are typically better \cite{rgbr,Abstraction2013}. 
Abstraction granularity, however, is restricted by computational resources as CFR requires space linear in the size of the abstract game. 

One approach to improving abstraction granularity is to partition a game into smaller pieces and run CFR on each piece, either independently \cite{AAMAS2010,StrategyGrafting,stitching} or concurrently \cite{stitching}. 
Strategies for each piece are referred to as \emph{experts} that during a match, only act when play reaches their piece of the game. 
The winner of the 2011 ACPC three-player instant runoff (IRO) event was an agent built with such expert strategies. 
Similar to Abou Risk and Szafron's \emph{heads-up experts}, the 2011 experts only acted in what appear to be the four most common two-player scenarios that resulted after one player had folded \cite[Table 4]{AAMAS2010}. 
In particular, an expert only acted after the opening sequence of player actions was $f$, $rf$, $rrf$, or $rcf$, where $f$ denotes the fold action, $c$ denotes call, and $r$ denotes raise. 
Two-player scenarios are convenient to work with since the elimination of one player greatly reduces the number of possible future action sequences and thus reduces the size of the game. 
These experts were computed independently using an abstraction with 169, 180,000, 540,000, and 78,480 buckets on each of the four betting rounds respectively. 
Here and throughout this section, the same $k$-means clustering technique on hand strength distribution from Section \ref{sec:holdem-experiments} was used to bucket hands and we refer the reader to Johanson \etal~\cite{Abstraction2013} for more details. 
To play the rest of the game, a \emph{base strategy} for the full, unpartitioned three-player game was computed with CFR using an abstraction with 169, 10,000, 5450, and 500 buckets per round respectively. 
Thus, the experts could distinguish between many more different hands compared to the base strategy, even though the abstract game for the base strategy still contained approximately 5.9 billion information sets. 
More details about this expert construction process are found in the description of the 2010 ACPC three-player IRO winner \cite{stitching} that was identical to the 2011 agent but used coarser abstractions. 
In 2009, the first year of the three-player events, the IRO winner was a simple base strategy computed with CFR in a very coarse abstraction \cite{AAMAS2010}.

\subsection{A New Three-Player Limit Texas Hold'em Agent}

As demonstrated above, partitioning a game into smaller pieces is a convenient method for increasing abstraction granularity. 
For the 2012 ACPC, we again used this same methodology to construct our new three-player limit Texas Hold'em agent. 
This time, rather than partitioning the game into special two-player scenarios, we partitioned the histories into two parts: an \emph{important} part and an \emph{unimportant} part. 
The important histories were defined as follows. 
First, we scanned all of the 2011 ACPC match logs that the winning IRO agent presented above played in and for each betting sequence, we calculated the frequency at which the agent was faced with a decision at that sequence. 
For example, the frequency the agent was faced with a decision at the empty betting sequence was $1/3$ since positions in the game rotate, making the agent first to act once in every three hands. 
Next, we multiplied each of these frequencies by the pot size at that betting sequence.  
For instance, we multiplied the $1/3$ frequency for the empty betting sequence by $15$ since the game is played with a small blind of $5$ chips and a big blind of $10$ chips, creating an initial pot of $15$ chips. 
For each history, if this value for the history's betting sequence was greater than $1/100$, then the history was labeled as important. 
Since $(1/3)\cdot 15 = 5 > 1/100$, the empty sequence was labeled as important. 
In addition, any prefix of an important history was also labeled as important, while the remaining histories were labeled as unimportant. 
Only $0.023$\% of the nonterminal betting sequences in three-player limit Hold'em belonged to the important part. 
While many of the important histories overlapped with the two-player scenarios used by the 2011 agent, there were several three-player scenarios, such as the empty sequence and the $rcc$ sequence, that were labeled important. 

Using this partition, we employed a very fine-grained abstraction on the important part and a coarse abstraction on the unimportant part. 
This way, our agent can distinguish between many more hands at the few sequences that historically were reached more frequently or that had lots of chips at stake. 
Our coarse abstraction for the unimportant part used the same 169, 1800, 200, and 50 buckets per round employed by our 2X abstraction in Section \ref{sec:holdem-experiments}, while our fine-grained abstraction for the important part used 169, 180,000, 765,000, and 840,000 buckets per round respectively. 
Strategies for both parts were computed concurrently \cite{stitching} across the 2.5 billion information set abstract game resulting from the two abstractions. 
Note that this abstract game is less than half the size of the abstract game used to compute the base strategy in 2011, meaning that less computer memory was required to run CFR. 
We used a parallel implementation of the External Sampling variant of CFR mentioned in Section \ref{sec:poker-games}, which ran for 16 days using 48 2.1 GHz AMD processors on a machine with 256GB of total RAM (though less than 100GB of RAM were needed).  

\subsection{Results}

\begin{table}[t]
\centering
\caption{Results of the 2012 ACPC three-player limit Hold'em events \cite{ACPC}. Earnings are in milli-big-blinds per game (mbb/g) and errors indicate 95\% confidence intervals.}
\begin{tabular}{r|c}
\multicolumn{2}{c}{\bf Total Bankroll} \\
\bf Agent & \bf Total Earnings \\
\hline
\cellcolor{yellow} Hyperborean3p & \cellcolor{yellow} $28 \pm 5$ \\
little.rock & $-4 \pm 7$ \\
neo.poker.lab & $-11 \pm 5$ \\
sartre & $-12 \pm 7$
\end{tabular}
\begin{tabular}{r|ccc}
\multicolumn{4}{c}{\bf Instant Runoff} \\
\bf Agent & \bf Round 1 & \bf Round 2 & \bf Round 3 \\
\hline
\cellcolor{yellow} Hyperborean3p & \cellcolor{yellow} $37 \pm 5$ & \cellcolor{yellow} $28 \pm 5$ & \cellcolor{yellow} $23 \pm 8$ \\
little.rock & $13 \pm 6$ & $-4 \pm 7$ & $-9 \pm 9$ \\
neo.poker.lab & $7 \pm 5$ & $-11 \pm 5$ & $-14 \pm 6$ \\
sartre & $5 \pm 7$ & $-12 \pm 7$ & Eliminated \\
dcubot & $-62 \pm 8$ & Eliminated & Eliminated
\end{tabular}
\label{tab:2012}
\end{table}

\begin{table}[t]
\centering
\caption{Results of a four-agent RRC between the ACPC IRO three-player winners from 2009, 2010, 2011, and 2012. Earnings are in milli-big-blinds per game for the row player against the column players and errors indicate 95\% confidence intervals.}
\begin{tabular}{r|ccccccc}
 & 09,10 & \small 09,11 & 09,12 & 10,11 & 10,12 & 11,12 & \bf Overall \\
\hline
\bf 2009 & - & - & - & $-21 \pm 7$ & $-26 \pm 5$ & $-31 \pm 5$ & \boldmath $-26 \pm 4$ \\
\bf 2010 & - & $0 \pm 4$ & $-5 \pm 3$ & - & - & $-23 \pm 5$ & \boldmath $-10 \pm 4$ \\
\bf 2011 & $21 \pm 6$ & - & $6 \pm 5$ & - & $8 \pm 5$ & - & \boldmath $11 \pm 4$ \\
\bf \cellcolor{yellow} 2012 & \cellcolor{yellow} $31 \pm 5$ & \cellcolor{yellow} $25 \pm 5$ & \cellcolor{yellow} - & \cellcolor{yellow} $16 \pm 4$ & \cellcolor{yellow} - & \cellcolor{yellow} - & \cellcolor{yellow} \boldmath $24 \pm 4$
\end{tabular}
\label{tab:acpc-summary}
\end{table}

The 2012 competition results \cite{ACPC} are presented in Table \ref{tab:2012}. 
Our 2012 agent, name \emph{Hyperborean3p}, won both the IRO and TBR events by significant margins. 
In addition, we compared our new agent against the previous IRO winners from the 2009, 2010, and 2011 competitions by running a four-agent round-robin competition (RRC). 
Table \ref{tab:acpc-summary} presents the results averaged across 10 RRCs. 
We see that not only does the 2012 agent require less computer memory to generate than the 2011 agent, the 2012 agent earns 13 milli-big-blinds per game more on average. 

Finally, all of the competition winners from 2009 to 2012 used the average strategy profiles generated by CFR. 
In light of our new insights from Section \ref{sec:theory} and as a final validation of our CFR modification, we reran CFR on the 2012 abstract game using the same CFR implementation on the same machine, except now saving the current profile and discarding the average. 
For several checkpoints of the original average strategy and the new current strategy, we played 10 RRCs versus the 2009, 2010, and 2011 ACPC IRO winners and plotted the results in Figure \ref{fig:rrc}. 
While the average strategy takes 20 days before earning 25 milli-big-blinds per game, the current strategy reaches better performance in just 5 days while requiring only half the memory (less than 50GB of RAM) to compute.

\section{Conclusion}
\label{sec:conclusion}

This paper provides the first theoretical advancements for applying CFR to games that are not two-player zero-sum. 
While previous work had demonstrated that CFR does not necessarily converge to a Nash equilibrium in such games, we have provided theoretical evidence that CFR eliminates iteratively strictly dominated actions and strategies. 
Thus, CFR provides a mechanism for removing iterative strict domination that was otherwise infeasible with previous techniques for large, non-zero-sum extensive-form games. 
In addition, our theory is the first step to understanding why CFR generates well-performing strategies in non-zero-sum games. 
Though our experiments show that the current profile reaches a high level of performance faster than the average, it remains unclear whether this is due to faster removal of domination that our theory illustrates. 
Nonetheless, we have shown that just using the current profile gives a more time and memory efficient implementation of CFR for games with more than two players that can lead to increased performance. 
Furthermore, we presented a new three-player limit Texas Hold'em agent that won both three-player events of the 2012 Annual Computer Poker Competition. 
Our agent uses a new partition of the game tree, requires less computer memory to generate than the 2011 winner, and outperforms the previous competition winners by a significant margin. 

\begin{figure}[t]%
\centering
\includegraphics[width=0.49\columnwidth]{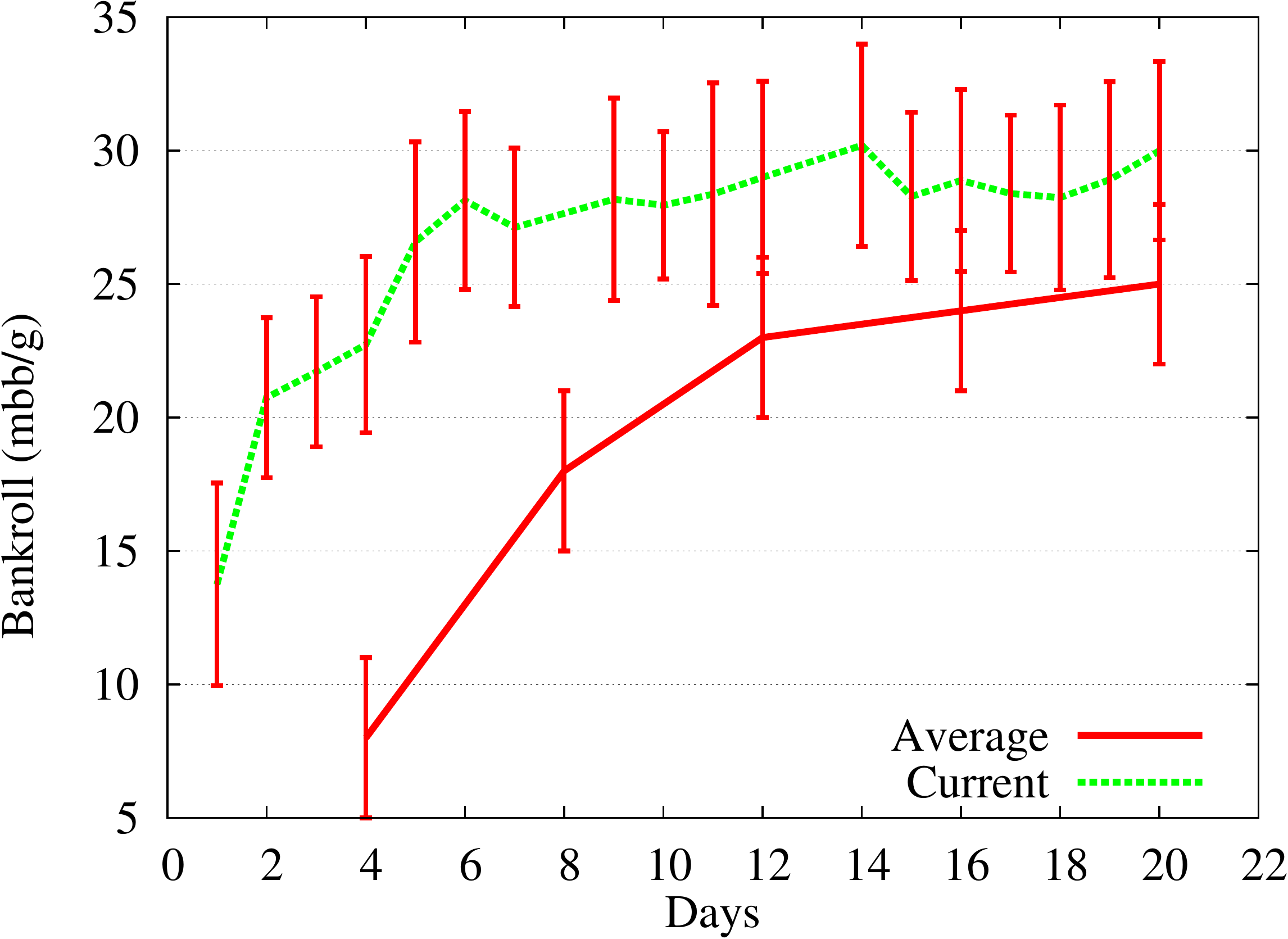}%
\caption{Performance over time (in days) of the average profile that won the three-player events of the 2012 ACPC, and of the current profile computed in the same abstract game. Error bars indicate 95\% confidence intervals over 10 competitions versus the top 2009, 2010, and 2011 ACPC IRO three-player agents.}%
\label{fig:rrc}%
\end{figure}

Future work will look at finding additional properties of CFR in non-zero-sum games that go beyond domination. 
Additionally, we would like to compare CFR's average and current profiles in other large, non-zero-sum domains outside of poker. 
Finally, this work has only considered the problem of computing strategies for play against a set of unknown opponents. 
In poker and other repeated games, we often gain information about the opponents' strategies over time. 
For repeated non-zero-sum games, using opponent modelling to adjust one's strategy could drastically improve play. 

\section*{Acknowledgements}

Thank-you to Martin Zinkevich and the members of the Computer Poker Research Group at the University of Alberta for helpful conversations pertaining to this research.  
This research was supported by Alberta Innovates -- Technology Futures and computing resources were provided by WestGrid and Compute Canada.

\section*{Vitae}

\begin{figure}[h]%
\centering
\includegraphics[width=0.49\columnwidth]{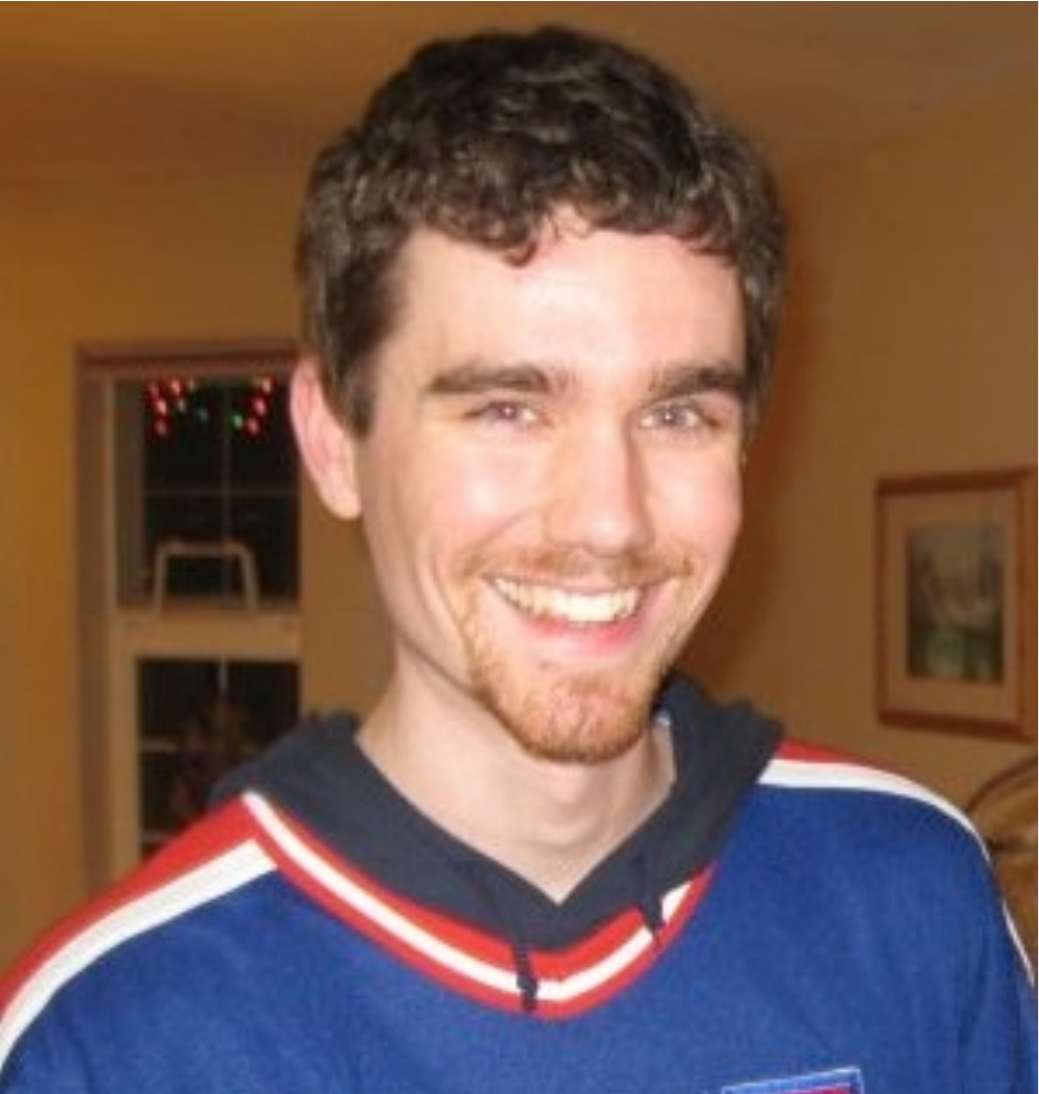}%
\label{fig:bio}%
\end{figure}

Richard Gibson is a Ph.D.~student in the Computing Science Department at the University of Alberta. 
He is a member of the Computer Poker Research Group and primary author of \emph{Hyperborean3p}, the reigning three-player limit Texas Hold'em champion of the Annual Computer Poker Competition. 
His research interests generally lie at the intersection of artificial intelligence and games. 

\appendix

\section{Proofs of Technical Results}
\label{sec:appendix}

In this appendix, we prove Proposition \ref{prop:weakdom} and Theorems \ref{thm:normaldom}, \ref{thm:domacts}, and \ref{thm:domstrats}. 
For $I \in \cI_i$, define
\[ D(I) = \{I' \in \cI_i \mid \exists h \in I, h' \in I' \text{ such that } h \sqsubseteq h' \} \] 
to be the set of information sets descending from $I$. 

\begin{prop-weakdom}
If $a$ is a weakly dominated action at $I \in \cI_i$ and $\sigma_i \in \Sigma_i$ satisfies $\pi_i^\sigma(I)\sigma_i(I,a) > 0$, then $\sigma_i$ is a weakly dominated strategy. 
\end{prop-weakdom}
\begin{proof} 
Since $a$ is weakly dominated, there exists a strategy $\sigma_i' \in \Sigma_i$ such that $v_i(I, \sigma_{(I \rightarrow a)}) \leq v_i(I, (\sigma_i', \sigma_{-i}))$ for all opponent profiles $\sigma_{-i} \in \Sigma_{-i}$, and there exists an opponent profile $\sigma_{-i}'$ such that $v_i(I, (\sigma_{i(I \rightarrow a)}, \sigma_{-i}')) < v_i(I, (\sigma_i', \sigma_{-i}'))$. 
Let $\hat{\sigma}_i$ be the strategy $\sigma_i$ except at $I$, where $\hat{\sigma}_i(I,a) = 0$ and $\hat{\sigma}_i(I,b) = \sigma_i(I,b) / (1 - \sigma_i(I,a))$ for all $b \in A(I)$, $b \neq a$. 
Next, for all $J \in \cI_i$ and $b \in A(J)$, define 
\[ \sigma_i''(J,b) = 
\left\{ 
\begin{array}{ll} 
\multirow{2}{*}{$\frac{\pi_i^\sigma(I) \left( \sigma_i(I,a) \pi_i^{\sigma'}(I,J) \sigma_i'(J,b) + ( 1 - \sigma_i(I,a) ) \pi_i^{\hat{\sigma}}(I,J) \hat{\sigma}_i(J,b) \right) }{\pi_i^\sigma(I) \left( \sigma_i(I,a) \pi_i^{\sigma'}(I,J) + ( 1 - \sigma_i(I,a) ) \pi_i^{\hat{\sigma}}(I,J) \right) }$} & \text{if } J \in D(I) \\
& \text{(and arbitrary when the} \\
& \text{ denominator is zero),} \\
\sigma_i(J,b) & \text{if } J \notin D(I).
\end{array} 
\right. 
\]
One can verify that $\sigma_i'' \in \Sigma_i$ is a valid strategy for player $i$. 
Now, fix $\sigma_{-i} \in \Sigma_{-i}$.  Then,
\begin{align*}
u_i(\sigma_i, \sigma_{-i}) &= \sum_{z \in Z_I} \pi^\sigma(z) u_i(z) + \sum_{z \notin Z_I} \pi^\sigma(z) u_i(z) \\
&= \pi_i^\sigma(I)\sum_{b \in A(I)} \sigma_i(I,b) v_i(I,\sigma_{(I \rightarrow b)}) + \sum_{z \notin Z_I} \pi^\sigma(z) u_i(z) \\
&\leq \pi_i^\sigma(I) \sigma_i(I,a) v_i(I, (\sigma_i', \sigma_{-i})) \\
&\ \ \ \ + \pi_i^\sigma(I)(1 - \sigma_i(I,a))\sum_{\substack{b \in A(I)\\b \neq a}} \hat{\sigma}_i(I,b) v_i(I, (\hat{\sigma}_{i(I \rightarrow b)}, \sigma_{-i})) \\
&\ \ \ \ + \sum_{z \notin Z_I} \pi^\sigma(z) u_i(z) \\
&= \pi_i^\sigma(I) v_i(I, (\sigma_i'', \sigma_{-i})) + \sum_{z \notin Z_I} \pi^\sigma(z) u_i(z) \\
&= u_i(\sigma_i'', \sigma_{-i}).
\end{align*}
Thus, $u_i(\sigma_i, \sigma_{-i}) \leq u_i(\sigma_i'', \sigma_{-i})$ for all $\sigma_{-i} \in \Sigma_{-i}$. 
A similar argument shows that $u_i(\sigma_i, \sigma_{-i}') < u_i(\sigma_i'', \sigma_{-i}')$, proving that $\sigma_i$ is weakly dominated by $\sigma_i''$.
\end{proof}

Next, we prove Theorem \ref{thm:normaldom}, using the fact that new iterative dominances only arise from removing actions and never from removing mixed strategies \cite{DominanceComplexity}:

\begin{theorem-normaldom}
Let $\sigma^1, \sigma^2, ...$ be a sequence of strategy profiles in a normal-form game where all players' strategies are computed by regret minimization algorithms where for all $i \in N$, $a \in A_i$, if $R_i^T(a) < 0$ and $R_i^T(a) < \max_{b \in A_i} R_i^T(b)$, then $\sigma_i^{T+1}(a) = 0$. 
If $\sigma_i$ is an iteratively strictly dominated strategy, then there exists an integer $T_0$ such that for all $T \geq T_0$, $\text{\emph{supp}}(\sigma_i) \nsubseteq \text{\emph{supp}}(\sigma_i^T)$. 
\end{theorem-normaldom}
\begin{proof} 
Let $a^1, a^2, ..., a^k$ be iteratively strictly dominated actions (pure strategies) for players $j_1, j_2, ..., j_k$ respectively that once removed in sequence yields strict domination of $\sigma_i$. 
Let $B_{-i} = A_{-i} \backslash \{a^1, a^2, ..., a^k\}$ be the set of opponent actions other than $a^1, a^2, ..., a^k$. 
Next, by iterative strict domination of $\sigma_i$ and because the game is finite, there exists another strategy $\sigma_i' \in \Sigma_i$ such that 
\[ \epsilon = \min_{a_{-i} \in B_{-i}} u_i(\sigma_i', a_{-i}) - u_i(\sigma_i, a_{-i}) > 0, \]
so that $u_i(\sigma_i, a_{-i}) \leq u_i(\sigma_i', a_{-i}) - \epsilon$ for all $a_{-i} \in B_{-i}$. 
Then, 
\begin{align}
\sum\limits_{a \in A_i}\sigma_i(a)R_i^T(a) &= \sum_{a \in A_i}\sigma_i(a)R_i^T(a) - \sum_{a \in A_i}\sigma_i'(a)R_i^T(a) + \sum_{a \in A_i}\sigma_i'(a)R_i^T(a) \nonumber \\
&= \sum_{a \in A_i} \left( \sigma_i(a) - \sigma_i'(a) \right) \sum_{t=1}^T \left( u_i(a, \sigma_{-i}^t) - u_i(\sigma^t) \right) \nonumber \\
&\ \ \ \  + \sum_{a \in A_i}\sigma_i'(a)R_i^T(a) \nonumber \\
&= \sum_{t=1}^{T} \left( u_i(\sigma_i, \sigma_{-i}^t) - u_i(\sigma_i', \sigma_{-i}^t) \right) + \sum_{a \in A_i}\sigma_i'(a)R_i^T(a) \nonumber \\
&= \sum_{\substack{\text{supp}(\sigma_{-i}^t) \nsubseteq B_{-i}\\1 \leq t \leq T}} \left( u_i(\sigma_i, \sigma_{-i}^t) - u_i(\sigma_i', \sigma_{-i}^t) \right) \nonumber \\
&\ \ \ \ + \sum_{\substack{\text{supp}(\sigma_{-i}^t) \subseteq B_{-i}\\1 \leq t \leq T}} \left( u_i(\sigma_i, \sigma_{-i}^t) - u_i(\sigma_i', \sigma_{-i}^t) \right) + \sum_{a \in A_i}\sigma_i'(a)R_i^T(a) \nonumber \\
&= \sum_{\substack{\text{supp}(\sigma_{-i}^t) \nsubseteq B_{-i}\\1 \leq t \leq T}} \left( u_i(\sigma_i, \sigma_{-i}^t) - u_i(\sigma_i', \sigma_{-i}^t) \right) \nonumber \\
&\ \ \ \ + \sum_{\substack{\text{supp}(\sigma_{-i}^t) \subseteq B_{-i}\\1 \leq t \leq T}} \sum_{a_{-i} \in B_{-i}} \sigma_{-i}^t(a_{-i}) \left( u_i(\sigma_i, a_{-i}) - u_i(\sigma_i', a_{-i}) \right) \nonumber \\
&\ \ \ \  + \sum_{a \in A_i}\sigma_i'(a)R_i^T(a), \text{ where } \sigma_{-i}(a_{-i}) = \prod_{j \neq i} \sigma_j(a_j) \nonumber \\
&\leq \sum_{\substack{\text{supp}(\sigma_{-i}^t) \nsubseteq B_{-i}\\1 \leq t \leq T}} \left( u_i(\sigma_i, \sigma_{-i}^t) - u_i(\sigma_i', \sigma_{-i}^t) \right) \nonumber \\
&\ \ \ \ + \sum_{\substack{\text{supp}(\sigma_{-i}^t) \subseteq B_{-i}\\1 \leq t \leq T}} (-\epsilon) + \max_{a \in A_i}R_i^T(a). \label{eq:normaldom}
\end{align}
We claim that there exists an integer $T_0$ such that for all $T \geq T_0$, there exists $a \in \text{supp}(\sigma_i)$ such that $R_i^T(a) < 0$ and $R_i^T(a) < \max_{b \in A_i} R_i^T(b)$.  
By our assumption, this implies that for all $T \geq T_0$, there exists an action $a \in \text{supp}(\sigma_i)$ such that $a \notin \text{supp}(\sigma_i^T)$, establishing the theorem. 

To complete the proof, it remains to establish the claim, which we prove by strong induction on $k$. 
For the base case $k = 0$, we have $B_{-i} = A_{-i}$, and so by equation \eqref{eq:normaldom} we have
\begin{align}
\min_{a \in \text{supp}(\sigma_i)} R_i^T(a) &\leq \sum\limits_{a \in A_i}\sigma_i(a)R_i^T(a) \nonumber \\
&\leq -\epsilon T + \max_{a \in A_i}R_i^T(a) \label{eq:normaldom2} \\
&\leq -\epsilon T + R_i^{T,+}. \nonumber
\end{align} 
Dividing both sides by $T$ and taking the limit superior gives
\begin{align*}
\limsup_{T \rightarrow \infty} \frac{1}{T} \min_{a \in \text{supp}(\sigma_i)} R_i^T(a) &\leq -\epsilon + \limsup_{T \rightarrow \infty} \frac{R_i^{T,+}}{T} \\
&= -\epsilon \\
&< 0.
\end{align*}
Thus, there exists an integer $T_0$ such that for all $T \geq T_0$, $R_i^T(a^*) < 0$ where $a^* = \text{arg}\min_{a \in \text{supp}(\sigma_i)} R_i^T(a)$. 
Also, by equation \eqref{eq:normaldom2}, $R_i^T(a^*) \leq -\epsilon T + \max_{a \in A_i} R_i^T(a) < \max_{a \in A_i} R_i^T(a)$, completing the base case. 

For the induction step, we may assume that there exist integers $T_1, ..., T_k$ such that for all $1 \leq \ell \leq k$, $T \geq T_\ell$, $R_{j_\ell}^T(a^\ell) < 0$ and $R_{j_\ell}^T(a^\ell) < \max_{b \in A_{j_\ell}} R_{j_\ell}^T(b)$. 
This means that for all $T \geq T_0' = \max\{T_1, ..., T_k\}$, $a^\ell \notin \text{supp}(\sigma_{j_\ell}^T)$ for all $1 \leq \ell \leq k$. 
Hence, $\text{supp}(\sigma_{-i}^T) \subseteq B_{-i}$ for all $T \geq T_0'$. 
Therefore, again setting $a^* = \text{arg}\min_{a \in \text{supp}(\sigma_i)}R_i^T(a)$, by equation \eqref{eq:normaldom} we have
\begin{align}
R_i^T(a^*) &\leq \sum_{a \in A_i} \sigma_i(a) R_i^T(a) \nonumber \\
&\leq \sum_{\substack{\text{supp}(\sigma_{-i}^t) \nsubseteq B_{-i}\\1 \leq t \leq T}} \left( u_i(\sigma_i, \sigma_{-i}^t) - u_i(\sigma_i', \sigma_{-i}^t) \right) \nonumber \\
&\ \ \ \ + \sum_{\substack{\text{supp}(\sigma_{-i}^t) \subseteq B_{-i}\\1 \leq t \leq T}} (-\epsilon) + \max_{a \in A_i}R_i^T(a) \nonumber \\
&\leq T_0' \Delta_i - \epsilon (T - T_0') + \max_{a \in A_i} R_i^T(a), \text{ where } \Delta_i = \max_{a,a' \in A} u_i(a) - u_i(a') \label{eq:normaldom3} \\
&\leq T_0' \Delta_i - \epsilon (T - T_0') + R_i^{T,+}. \nonumber
\end{align}
Dividing both sides by $T$ and taking the limit superior gives
\begin{align*}
\limsup_{T \rightarrow \infty} \frac{R_i^T(a^*)}{T} &\leq \limsup_{T \rightarrow \infty} \left( \frac{T_0' \Delta_i}{T} - \frac{\epsilon( T - T_0' ) }{T} + \frac{R_i^{T,+}}{T} \right) \\
&= -\epsilon \\
&< 0. 
\end{align*}
Thus, there exists an integer $T_0$ such that for all $T \geq T_0$, $T_0' \Delta_i < \epsilon(T - T_0')$ and $R_i^T(a^*) < 0$. 
By equation \eqref{eq:normaldom3}, this also means that for $T \geq T_0$, $R_i^T(a^*) < \max_{a \in A_i} R_i^T(a)$, completing the induction step. 
This establishes the claim and completes the proof.
\end{proof}

Before proving Theorems \ref{thm:domacts} and \ref{thm:domstrats}, we need an additional lemma. 
For $\sigma_i \in \Sigma_i$ and $I \in \cI_i$, define the \defword{full counterfactual regret} for $\sigma_i$ at $I$ to be 
\[ R_{i,\text{full}}^T(I, \sigma_i) = \sum_{t=1}^T (v_i(I, (\sigma_i, \sigma_{-i}^t)) - v_i(I, \sigma^t)). \] 
We begin by relating full counterfactual regret to a sum over cumulative counterfactual regrets. 
This step was part of the original CFR analysis \cite{CfrTechReport}, but we relate these terms here in a slightly different form. 
For $I, I' \in \cI_i$, $h \in I$, $h' \in I'$, and $\sigma_i \in \Sigma_i$, define $\pi_i^\sigma(I, I') = \pi_i(h, h')$, which is well-defined due to perfect recall. 
\begin{lemma}
\label{lem:fullcfr}
\[ R_{i,\text{\emph{full}}}^T(I, \sigma_i) = \sum_{I' \in D(I)} \pi_i^\sigma(I, I') \sum_{a \in A(I')} \sigma_i(I',a) R_i^T(I',a). \]
\end{lemma}
\begin{proof}
We prove the lemma by strong induction on $|D(I)|$. For $I \in \cI_i$ and $a \in A(I)$, define 
\begin{align*} 
S(I,a) = \{I' \in \cI_i \mid &\exists h \in I, h' \in I' \text{ where } ha \sqsubseteq h' \\ &\text{ and } \nexists h'' \in H_i \text{ where } ha \sqsubseteq h'' \sqsubset h' \}
\end{align*}
to be the set of all possible successor information sets for player $i$ after taking action $a$ at $I$. In addition, define $Z(I,a)$ to be the set of terminal histories where the last action taken by player $i$ was $a$ at $I$. 
To begin,
\begin{align}
R_{i,\text{full}}^T(I, \sigma_i) &= \sum_{t=1}^T v_i(I, (\sigma_i, \sigma_{-i}^t)) - \sum_{t=1}^T v_i(I, \sigma^t) \nonumber \\
&= \sum_{t=1}^T \sum_{a \in A(I)} \sigma_i(I,a) v_i(I, (\sigma_{i(I \rightarrow a)}, \sigma_{-i}^t)) - \sum_{t=1}^T v_i(I, \sigma^t) \nonumber \\
&= \sum_{a \in A(I)} \sigma_i(I,a) \sum_{t=1}^T \left( \sum_{z \in Z(I,a)} \pi_{-i}^{\sigma^t}(z)u_i(z) \right. \nonumber \\
&\ \ \ \  \left. + \sum_{I' \in S(I,a)} v_i(I', (\sigma_i, \sigma_{-i}^t)) \right) - \sum_{t=1}^T v_i(I, \sigma^t). \label{eq:domactslem}
\end{align}
For the base case $D(I) = \{I\}$, we have $S(I,a) = \emptyset$ and $Z(I,a) = Z_I$, and so the right hand side of equation \eqref{eq:domactslem} reduces to $\sum_{a \in A(I)} \sigma_i(I,a) R_i^T(I,a)$ as desired. 
For the induction step, note that $|D(I')| < |D(I)|$ for all $I' \in S(I,a)$, and so we may apply the induction hypothesis to get, for all $I' \in S(I,a)$,
\begin{align*}
\sum_{t=1}^T v_i(I', (\sigma_i, \sigma_{-i}^t)) &= R_{i,\text{full}}^T(I', \sigma_i) + \sum_{t=1}^T v_i(I', \sigma^t) \\
&= \sum_{I'' \in D(I')} \pi_i^\sigma(I', I'') \sum_{b \in A(I'')} \sigma(I'', b) R_i^{T}(I'', b) \\
&\ \ \ \ + \sum_{t=1}^T v_i(I', \sigma^t).
\end{align*}
Finally, substituting into equation \eqref{eq:domactslem}, we have
\begin{align*}
R_{i,\text{full}}^T(I, \sigma_i) &= \sum_{a \in A(I)} \sigma_i(I,a) \left[ \sum_{t=1}^T \sum_{z \in Z(I,a)} \pi_{-i}^{\sigma^t}(z)u_i(z) \right. \\
&\ \ \ \  + \sum_{I' \in S(I,a)} \left( \sum_{I'' \in D(I')} \pi_i^\sigma(I', I'') \sum_{b \in A(I'')} \sigma_i(I'', b) R_i^{T}(I'', b) \right. \\
&\ \ \ \  \left. \left. + \sum_{t=1}^T v_i(I', \sigma^t) \right) \right] - \sum_{t=1}^T v_i(I, \sigma^t) \\
&= \sum_{a \in A(I)} \sigma_i(I,a) \sum_{t=1}^T v_i(I, \sigma^t_{(I \rightarrow a)}) - \sum_{t=1}^T v_i(I, \sigma^t) \\
&\ \ \ \  + \sum_{a \in A(I)} \sigma_i(I,a) \sum_{I' \in S(I,a)} \left( \sum_{I'' \in D(I')} \pi_i^\sigma(I', I'') \sum_{b \in A(I'')} \sigma_i(I'', b) R_i^{T}(I'', b) \right) \\
&= \sum_{a \in A(I)} \sigma_i(I,a) R_i^T(I, a) \\
&\ \ \ \  + \sum_{\substack{I' \in D(I)\\I' \neq I}} \pi_i^\sigma(I, I') \sum_{b \in A(I')} \sigma_i(I', b) R_i^{T}(I', b) \\
&= \sum_{I' \in D(I)} \pi_i^\sigma(I, I') \sum_{a \in A(I')} \sigma_i(I', a) R_i^{T}(I', a),
\end{align*}
completing the proof. 
\end{proof}

\begin{corollary}
\label{cor:minfullcfr}
\[ R_{i,\text{\emph{full}}}^T(I, \sigma_i) \leq \Delta_i |D(I)| \sqrt{|A(\cI_i)|T}. \]
\end{corollary}
\begin{proof}
By Lemma \ref{lem:fullcfr}, 
\begin{align*}
R_{i,\text{full}}^T(I, \sigma_i) &= \sum_{I' \in D(I)} \pi_i^\sigma(I, I') \sum_{a \in A(I')} \sigma_i(I', a) R_i^{T}(I', a) \\
&\leq \sum_{I' \in D(I)} \max_{a \in A(I)} R_i^{T,+}(I', a) \\
&\leq |D(I)| \Delta_i \sqrt{|A(\cI_i)|T}
\end{align*}
by equation \eqref{eq:RMBound}.
\end{proof}

\begin{theorem-domacts}
Let $\sigma^1, \sigma^2, ...$ be strategy profiles generated by CFR in an extensive-form game, let $I \in \cI_i$, and let $a$ be an iteratively strictly dominated action at $I$, where removal in sequence of the iteratively strictly dominated actions $a_1, ..., a_k$ at $I_1, ..., I_k$ respectively yields iterative dominance of $a_{k+1} = a$. 
If for $1 \leq \ell \leq k+1$, there exist real numbers $\delta_\ell, \gamma_\ell > 0$ and an integer $T_\ell$ such that for all $T \geq T_\ell$, $| \Sigma_{\delta_\ell}(I_\ell) \cap \{ \sigma^t \mid T_\ell \leq t \leq T \} | \geq \gamma_\ell T$, then
\begin{itemize}
\item[\emph{(i)}] there exists an integer $T_0$ such that for all $T \geq T_0$, $R_i^T(I,a) < 0$,
\item[\emph{(ii)}] if $\lim_{T \rightarrow \infty} x^T / T = 0$, then $\lim_{T \rightarrow \infty} y^T(I,a) / T = 0$, where $y^T(I,a)$ is the number of iterations $1 \leq t \leq T$ satisfying $\sigma^t(I,a) > 0$, and
\item[\emph{(iii)}] if $\lim_{T \rightarrow \infty} x^T / T = 0$, then $\lim_{T \rightarrow \infty} \pi_i^{\bar{\sigma}^T}(I) \bar{\sigma}_i^T(I,a) = 0$. 
\end{itemize}
\end{theorem-domacts}
\begin{proof} 
We will first prove parts (i) and (ii) by strong induction on $k$, followed by proving (iii) from (ii). 
For $\delta \geq 0$, let $\hat{\Sigma}_{\delta}(I) = \{ \sigma \in \Sigma_{\delta}(I) \mid \sigma(I_\ell, a_\ell) = 0, 1 \leq \ell \leq k \}$ be the set of strategies in $\Sigma_\delta(I)$ that do not play $a_1, ..., a_k$. 
By iterative strict domination of $a$, there exists $\sigma_i' \in \Sigma_i$ such that $v_i(I, \sigma_{(I \rightarrow a)}) \leq v_i(I, (\sigma_i', \sigma_{-i}))$ for all $\sigma \in \hat{\Sigma}_0(I)$. 
Next, let $\delta = \delta_{k+1}$ and $\gamma = \gamma_{k+1}$. 
Then, since $\hat{\Sigma}_\delta(I)$ is a closed and bounded set and $v_i(I, \cdot)$ is continuous, by the Balzano-Weierstrass Theorem there exists an $\epsilon > 0$ such that $v_i(I, \sigma_{(I \rightarrow a)}) \leq v_i(I, (\sigma_i', \sigma_{-i})) - \epsilon$ for all $\sigma \in \hat{\Sigma}_\delta(I)$. 
Then, 
\begin{align}
R_i^T(I,a) &= R_i^T(I,a) - R_{i,\text{full}}^T(I, \sigma_i') + R_{i,\text{full}}^T(I, \sigma_i') \nonumber \\
&= \sum_{t=1}^T \left( v_i(I, \sigma_{(I \rightarrow a)}^t) - v_i(I, (\sigma_i', \sigma_{-i}^t)) \right) + R_{i,\text{full}}^T(I, \sigma_i') \nonumber \\
&= \sum_{t=1}^{T_0' - 1} \left( v_i(I, \sigma_{(I \rightarrow a)}^t) - v_i(I, (\sigma_i', \sigma_{-i}^t)) \right) \nonumber \\
&\ \ \ \ + \sum_{\substack{T_0' \leq t \leq T\\\sigma^t \notin \hat{\Sigma}_0(I)}} \left( v_i(I, \sigma_{(I \rightarrow a)}^t) - v_i(I, (\sigma_i', \sigma_{-i}^t)) \right) \nonumber \\
&\ \ \ \ + \sum_{\substack{T_0' \leq t \leq T\\\sigma^t \in \hat{\Sigma}_\delta(I)}} \left( v_i(I, \sigma_{(I \rightarrow a)}^t) - v_i(I, (\sigma_i', \sigma_{-i}^t)) \right) \nonumber \\ 
&\ \ \ \ + \sum_{\substack{T_0' \leq t \leq T\\\sigma^t \in \hat{\Sigma}_0(I) \backslash \hat{\Sigma}_\delta(I)}} \left( v_i(I, \sigma_{(I \rightarrow a)}^t) - v_i(I, (\sigma_i', \sigma_{-i}^t)) \right) + R_{i,\text{full}}^T(I, \sigma_i') \label{eq:domacts}. 
\end{align}
For the base case $k = 0$, we have $\hat{\Sigma}_0(I) = \Sigma$ and $\hat{\Sigma}_\delta(I) = \Sigma_\delta(I)$. 
Choose $T_0$ to be any integer greater than $\max\{ T_0', \Delta_i^2 |D(I)|^2 |A(\cI_i)| / \epsilon^2 \gamma^2 \}$ so that for all $T \geq T_0$, 
\begin{align*}
R_i^T(I,a) &= \sum_{t=1}^{T_0' - 1} \left( v_i(I, \sigma_{(I \rightarrow a)}^t) - v_i(I, (\sigma_i', \sigma_{-i}^t)) \right) \\
&\ \ \ \ + \sum_{\substack{T_0' \leq t \leq T\\\sigma^t \in \Sigma_\delta(I)}} \left( v_i(I, \sigma_{(I \rightarrow a)}^t) - v_i(I, (\sigma_i', \sigma_{-i}^t)) \right) \nonumber \\ 
&\ \ \ \ + \sum_{\substack{T_0' \leq t \leq T\\\sigma^t \notin \Sigma_\delta(I)}} \left( v_i(I, \sigma_{(I \rightarrow a)}^t) - v_i(I, (\sigma_i', \sigma_{-i}^t)) \right) + R_{i,\text{full}}^T(I, \sigma_i') \\
&\leq -\epsilon |\Sigma_\delta(I) \cap \{\sigma^t \mid T_0 \leq t \leq T\}| + R_{i,\text{full}}^T(I, \sigma_i') \\
&\leq -\epsilon \gamma T + \Delta_i |D(I)| \sqrt{|A(\cI_i)|T} \text{ by Corollary \ref{cor:minfullcfr}} \\
&< 0
\end{align*}
by choice of $T_0$. 
This establishes part (i) of the base case. 
For part (ii), since CFR applies regret matching at $I$, by equation \eqref{eq:CFRRM}
it follows that for all $T \geq T_0$, either $\sum_{b \in A(I)} R_i^{T,+}(I,b) = 0$ or $\sigma_i^{T+1}(I,a) = 0$. 
Thus,
\begin{align*}
\lim_{T \rightarrow \infty} \frac{y^T(I,a)}{T} &= \lim_{T \rightarrow \infty} \frac{y^{T_0}(I,a) + (y^T(I,a) - y^{T_0}(I,a))}{T} \\
&\leq \lim_{T \rightarrow \infty} \frac{y^{T_0}(I,a) + x^T}{T} \\
&= 0.
\end{align*}
Thus, (ii) holds and we have established the base case of our induction. 

For the induction step, we now assume that parts (i) and (ii) hold for all $a_1, ..., a_k$. 
We will show that there exists an integer $T_0$ such that for all $T \geq T_0$, $R_i^T(I,a) < 0$.
This will establish part (i), and part (ii) will then follow as before to complete the induction step.  

Firstly, note that
\[ \sum_{\substack{T_0' \leq t \leq T\\\sigma^t \in \hat{\Sigma}_0(I) \backslash \hat{\Sigma}_\delta(I)}} \left( v_i(I, \sigma_{(I \rightarrow a)}^t) - v_i(I, (\sigma_i', \sigma_{-i}^t)) \right) \leq 0 \]
by iterative domination of $a$. 
Secondly, \\
$\mathlarger{\sum}_{\substack{T_0' \leq t \leq T\\\sigma^t \in \hat{\Sigma}_\delta(I)}} \left( v_i(I, \sigma_{(I \rightarrow a)}^t) - v_i(I, (\sigma_i', \sigma_{-i}^t)) \right)$
\begin{align*}
&\leq -\epsilon |\hat{\Sigma}_\delta(I) \cap \{\sigma^t \mid T_0 \leq t \leq T\}| \\
&= -\epsilon \left( |\Sigma_\delta(I) \cap \{\sigma^t \mid T_0 \leq t \leq T\}| - |(\Sigma_\delta(I) \backslash \hat{\Sigma}_\delta(I)) \cap \{\sigma^t \mid T_0 \leq t \leq T\}| \right) \\
&\leq -\epsilon \gamma T + \epsilon \sum_{\ell = 1}^k y^T(I_\ell, a_\ell). 
\end{align*}
Thirdly, \\
\[ \sum_{\substack{T_0' \leq t \leq T\\\sigma^t \notin \hat{\Sigma}_0(I)}} \left( v_i(I, \sigma_{(I \rightarrow a)}^t) - v_i(I, (\sigma_i', \sigma_{-i}^t)) \right) \leq \Delta_i\sum_{\ell = 1}^k y^T(I_\ell, a_\ell). \]
Thus, substituting these three inequalities and Corollary \ref{cor:minfullcfr} into equation \eqref{eq:domacts} gives 
\begin{align*}
R_i^T(I,a) &\leq \sum_{t=1}^{T_0' - 1} \left( v_i(I, \sigma_{(I \rightarrow a)}^t) - v_i(I, (\sigma_i', \sigma_{-i}^t)) \right) \\
&\ \ \ \ + \Delta_i\sum_{\ell = 1}^k y^T(I_\ell, a_\ell) -\epsilon \gamma T + \epsilon \sum_{\ell = 1}^k y^T(I_\ell, a_\ell) + \Delta_i|D(I)|\sqrt{|A(\cI_i)|T}.
\end{align*}
Dividing both sides by $T$ and taking the limit superior gives 
\begin{align*}
\limsup_{T \rightarrow \infty} \frac{R_i^T(I,a)}{T} &\leq \sum_{t=1}^{T_0' - 1} \left( v_i(I, \sigma_{(I \rightarrow a)}^t) - v_i(I, (\sigma_i', \sigma_{-i}^t)) \right) \limsup_{T \rightarrow \infty} \frac{1}{T} \\
&\ \ \ \ \ + (\Delta_i + \epsilon) \sum_{\ell = 1}^k \limsup_{T \rightarrow \infty} \frac{y^T(I_\ell,a_\ell)}{T} - \epsilon \gamma + \Delta_i|D(I)|\sqrt{|A(\cI_i)|} \limsup_{T \rightarrow \infty} \frac{1}{\sqrt{T}} \\
&= -\epsilon \gamma \\
&< 0
\end{align*}
by applying part (ii) of the induction hypothesis. 
Therefore, there exists an integer $T_0$ such that for all $T \geq T_0$, $R_i^T(I,a) / T < 0$ and thus $R_i^T(I,a) < 0$, completing the induction step. 

Parts (i) and (ii) are now proven. 
It remains to prove (iii). 
To that end,
\begin{align*}
\lim_{T \rightarrow \infty} \pi_i^{\bar{\sigma}^T}(I) \bar{\sigma}_i^T(I,a) &= \lim_{T \rightarrow \infty} \left( \frac{1}{T} \sum_{t=1}^T \pi_i^{\sigma^t}(I) \right) \frac{\sum_{t=1}^T \pi_i^{\sigma^t}(I) \sigma_i^t(I,a)}{\sum_{t=1}^T \pi_i^{\sigma^t}(I)} \\
&= \lim_{T \rightarrow \infty} \frac{\sum_{t=1}^T \pi_i^{\sigma^t}(I) \sigma_i^t(I,a)}{T} \\
&\leq \lim_{T \rightarrow \infty} \frac{y^T(I,a)}{T} \\
&= 0
\end{align*}
by part (ii). 
Since $\pi_i^{\bar{\sigma}^T}(I) \bar{\sigma}_i^T(I,a)$ is nonnegative, it follows that $\lim_{T \rightarrow \infty} \pi_i^{\bar{\sigma}^T}(I) \bar{\sigma}_i^T(I,a) = 0$, completing the proof.
\end{proof}

\begin{theorem-domstrats}
Let $\sigma^1, \sigma^2, ...$ be strategy profiles generated by CFR in an extensive-form game, and let $\sigma_i$ be an iteratively strictly dominated strategy. Then,
\begin{itemize}
\item[\emph{(i)}] there exists an integer $T_0$ such that for all $T \geq T_0$, there exist $I \in \cI_i$, $a \in A(I)$ such that $\pi_i^\sigma(I)\sigma_i(I,a) > 0$ and $R_i^T(I,a) < 0$, and
\item[\emph{(ii)}] if $\lim_{T \rightarrow \infty} x^T / T = 0$, then $\lim_{T \rightarrow \infty} y^T(\sigma_i) / T = 0$, where $y^T(\sigma_i)$ is the number of iterations $1 \leq t \leq T$ satisfying $\text{supp}(\sigma_i) \subseteq \text{supp}(\sigma_i^t)$.
\end{itemize} 
\end{theorem-domstrats}
\begin{proof} 
Let $s_{j_1}^1, s_{j_2}^2, ..., s_{j_k}^k$ be iteratively strictly dominated pure strategies that once removed in sequence yields strict domination of $\sigma_i$. 
Let $S_{-i} = \cS_{-i} \backslash \{s_{j_1}^1, s_{j_2}^2, ..., s_{j_k}^k\}$ be the set of opponent pure strategy profiles that do not play any of $s_{j_1}^1, s_{j_2}^2, ..., s_{j_k}^k$. 
Next, by iterative strict domination of $\sigma_i$ and because the game is finite, there exists another strategy $\sigma_i' \in \Sigma_i$ such that 
\[ \epsilon = \min_{s_{-i} \in S_{-i}} u_i(\sigma_i', s_{-i}) - u_i(\sigma_i, s_{-i}) > 0,\] 
so that $u_i(\sigma_i, s_{-i}) \leq u_i(\sigma_i', s_{-i}) - \epsilon$ for all $s_{-i} \in S_{-i}$. 

For $\hat{\sigma}_i \in \Sigma_i$, define $R_{i,\text{full}}^T(\hat{\sigma_i}) = \sum_{t=1}^{T} \left( u_i(\hat{\sigma_i}, \sigma_{-i}^t) - u_i(\sigma^t) \right)$. 
Note that 
\begin{align*}
R_{i,\text{full}}^T(\hat{\sigma_i}) &= \sum_{I \in \hat{\cI}_i} R_{i,\text{full}}^T(I, \hat{\sigma_i}), 
\end{align*}
where $\hat{\cI}_i = \{I \in \cI_i \mid \forall h \in I, h' \sqsubset h, P(h') \neq i \}$ is the set of all possible first information sets for player $i$ reached. 
So, by Corollary \ref{cor:minfullcfr}, $R_{i,\text{full}}^T(\hat{\sigma_i}) \leq \Delta_i |\cI_i| \sqrt{|A(\cI_i)|T}$ for all $\hat{\sigma}_i \in \Sigma_i$. 
Then by Lemma \ref{lem:fullcfr}, we have \\
$\mathlarger{\sum}_{I \in \cI_i} \pi_i^\sigma(I) \mathlarger{\sum}_{a \in A(I)} \sigma_i(I,a) R_i^T(I,a)$
\begin{align}
&= R_{i,\text{full}}^T(\sigma_i) - R_{i,\text{full}}^T(\sigma_i') + R_{i,\text{full}}^T(\sigma_i') \nonumber \\
&= \sum_{t=1}^{T} \left( u_i(\sigma_i, \sigma_{-i}^t) - u_i(\sigma_i', \sigma_{-i}^t) \right) + R_{i,\text{full}}^T(\sigma_i') \nonumber \\
&= \sum_{\substack{\text{supp}(\sigma_{-i}^t) \subseteq S_{-i}\\1 \leq t \leq T}} \sum_{s_{-i} \in S_{-i}} \sigma_{-i}^t(s_{-i}) \left( u_i(\sigma_i, s_{-i}) - u_i(\sigma_i', s_{-i}) \right) \nonumber \\
&\ \ \ \ + \sum_{\substack{\text{supp}(\sigma_{-i}^t) \nsubseteq S_{-i}\\1 \leq t \leq T}} \left( u_i(\sigma_i, \sigma_{-i}^t) - u_i(\sigma_i', \sigma_{-i}^t) \right) + R_{i,\text{full}}^T(\sigma_i'), \nonumber \\
&\ \ \ \ \text{ where } \sigma_{-i}(s_{-i}) = \prod_{\substack{j \neq i\\I \in \cI_j}} \sigma_j(I, s_j(I)) \nonumber \\
&\leq -\epsilon \left( T - \sum_{\ell=1}^k y^T(s_{j_\ell}^\ell) \right) + \Delta_i \sum_{\ell = 1}^k y^T(s_{j_\ell}^\ell) + \Delta_i |\cI_i| \sqrt{|A(\cI_i)|T}. \label{eq:domstrats}
\end{align}
We claim that 
\[ \limsup_{T \rightarrow \infty} \frac{1}{T}\sum_{I \in \cI_i} \pi_i^\sigma(I) \mathlarger{\sum}_{a \in A(I)} \sigma_i(I,a) R_i^T(I,a) < 0. \] 
Assuming the claim holds, because $(1/T)$, $\pi_i^\sigma(I)$, and $\sigma_i(I,a)$ are nonnegative, it follows that there exists an integer $T_0$ such that for all $T \geq T_0$, there exist $I \in \cI_i$, $a \in A(I)$ such that $\pi_i^\sigma(I) \sigma_i(I,a) > 0$ and $R_i^T(I,a) < 0$, establishing (i). 
For part (ii), note that part (i) and equation \eqref{eq:CFRRM} 
imply that for all $T \geq T_0$, either $\sum_{b \in A(I)} R_i^{T,+}(I,b) = 0$ or $\text{supp}(\sigma_i) \nsubseteq \text{supp}(\sigma_i^T)$. 
Thus, 
\begin{align*}
\lim_{T \rightarrow \infty} \frac{y^T(\sigma_i)}{T} &= \lim_{T \rightarrow \infty} \frac{y^{T_0}(\sigma_i) + (y^T(\sigma_i) - y^{T_0}(\sigma_i))}{T} \\
&\leq \lim_{T \rightarrow \infty} \frac{y^{T_0}(\sigma_i) + x^T}{T} \\
&= 0,
\end{align*}
establishing part (ii). 

To complete the proof, it remains to prove the claim, which we will prove by induction on $k$. 
For the base case $k = 0$, equation \eqref{eq:domstrats} gives 
\begin{align*}
\limsup_{T \rightarrow \infty} \frac{1}{T} \sum_{I \in \cI_i} \pi_i^\sigma(I) \mathlarger{\sum}_{a \in A(I)} \sigma_i(I,a) R_i^T(I,a) &\leq \limsup_{T \rightarrow \infty} -\epsilon  + \frac{\Delta_i |\cI_i| \sqrt{|A(\cI_i)|}}{\sqrt{T}} \\
&= -\epsilon \\
&< 0.
\end{align*}
For the induction step, we may assume that parts (i) and (ii) hold for all $s_{j_1}^1, s_{j_2}^2, ..., s_{j_k}^k$. 
Then equation \eqref{eq:domstrats} implies
\begin{align*}
\limsup_{T \rightarrow \infty} \frac{1}{T} \sum_{I \in \cI_i} \pi_i^\sigma(I) \sum_{a \in A(I)} \sigma_i(I,a) R_i^T(I,a) &\leq -\epsilon + (\epsilon + \Delta_i) \sum_{\ell = 1}^k \limsup_{T \rightarrow \infty} \frac{y^T(s_{j_\ell}^\ell)}{T} \\
&\ \ \ \ + \limsup_{T \rightarrow \infty} \frac{\Delta_i |\cI_i| \sqrt{|A(\cI_i)|}}{\sqrt{T}} \\
&= -\epsilon \\
&< 0,
\end{align*}
proving the claim. 
\end{proof}



\bibliographystyle{model1-num-names}
\bibliography{citations}








\end{document}